\documentclass[11pt,reqno]{amsart}
\usepackage{amsfonts, amsmath, amssymb, amsthm, mathtools, mathrsfs, dsfont, tikz-cd, xcolor, braket}

\usepackage{geometry}
\usepackage[linktocpage, colorlinks=true, linkcolor={blue}, citecolor={blue}]{hyperref}

\usepackage[capitalise]{cleveref}
\crefname{theorem}{Theorem}{Theorems}
\crefname{lemma}{Lemma}{Lemmas}
\crefname{proposition}{Proposition}{Propositions}

\usepackage{enumitem}

\theoremstyle{definition}
\newtheorem{theorem}{Theorem}[section]
\newtheorem{corollary}[theorem]{Corollary}
\newtheorem{lemma}[theorem]{Lemma}
\newtheorem{proposition}[theorem]{Proposition}
\newtheorem{definition}[theorem]{Definition}
\newtheorem{remark}[theorem]{Remark}
\newtheorem{example}[theorem]{Example}

\newcommand{\mb}{\mathbb}
\newcommand{\mc}{\mathcal}
\newcommand{\mf}{\mathfrak}
\newcommand{\lset}{\left\{ }
\newcommand{\rset}{\right\} }

\newcommand{\lpara}{\left(}
\newcommand{\rpara}{\right)}
\newcommand{\lbra}{\left[}
\newcommand{\rbra}{\right]}
\newcommand{\inv}{^{-1}}
\newcommand{\seminorm}[1]{{\left\vert\kern-0.25ex\left\vert\kern-0.25ex\left\vert #1 \right\vert\kern-0.25ex\right\vert\kern-0.25ex\right\vert}}

    \DeclarePairedDelimiter\abs{\lvert}{\rvert}%
    \DeclarePairedDelimiter\norm{\lVert}{\rVert}%
    \makeatletter
    \let\oldabs\abs
    \def\abs{\@ifstar{\oldabs}{\oldabs*}}
    \let\oldnorm\norm
    \def\norm{\@ifstar{\oldnorm}{\oldnorm*}}
    \makeatother

\newcommand{\lan}{\langle}
\newcommand{\ran}{\rangle}
\newcommand{\kla}{\left[}
\newcommand{\mer}{\right]}
\newcommand{\ten}{\otimes}

\newcommand{\kl}{\leq}

\DeclareMathOperator{\tr}{tr}
\DeclareMathOperator{\id}{id}

\DeclareMathOperator{\supp}{supp}
\DeclareMathOperator{\spn}{span}
\DeclareMathOperator{\ad}{ad}
\DeclareMathOperator{\Ad}{Ad}

\DeclareMathOperator{\chan}{Channel}
\DeclareMathOperator{\ev}{ev}
\DeclareMathOperator{\diag}{diag}

\title{How Far do Lindbladians Go?}
\author{Jihong Cai}
\email{jihongc2@illinois.edu}
\author{Advith Govindarajan}
\email{agovind2@illinois.edu}
\author{Marius Junge}
\email{mjunge@illinois.edu}
\thanks{MJ is partially supported by NSF-DMS 2247114}
\address{Department of Mathematics, University of Illinois Urbana-Champaign}

\begin{document}
\maketitle
\begin{abstract}
    We study controllability of finite-dimensional open quantum systems under a general Markovian control model combining full coherent (unitary) control with tunable dissipative channels. Assuming the Hamiltonian controls is a Hörmander system that generate $\mf{su}(n)$, we ask how little dissipation suffices to make the full state space $\mc D(\mc H)$ controllable. We show that minimal non-unital noise can break unitary-orbit invariants and, in many cases, a very small set of jump operators yields transitivity on $\mc D(\mc H)$. For multi-qubit systems we prove explicit transitivity results for natural resources such as a single-qubit amplitude-damping jump together with a dephasing channel, and we identify obstructions when only self-adjoint jump operators are available (yielding only unital evolutions). 

    We further develop a geometric viewpoint and ask the ``lifting'' question: when can a path of densities be obtained from applying a time-dependent family of Lindbladian to an initial state? For this, we have to analyze the tangent structure of the ``manifold with corners'' and how this tangent structure reflects Lindbldian evolution. Building on this framework, we derive reachability criteria and no-go results based on a norm-decrease alignment condition, including a geometric obstruction arising from the incompatibility between admissible tangent directions and dissipative contraction.
\end{abstract}

\bigbreak
\setcounter{tocdepth}{2}
\tableofcontents
\vfill\newpage

\section{Introduction}
The dynamics of a finite-dimensional open quantum system under control are naturally described by a Markovian master equation
$$\dot\rho_t =-i\Bigl[H_0+\sum_j u_j(t) H_j,\rho_t\Bigr]+\sum_k \gamma_k(t)\,L_k(\rho_t),$$
where $\rho_t\in\mc D(\mc H)$, the Hamiltonian controls $\{H_j\}$ generate coherent motion, and the Lindblad generators $\{L_k\}$ describe dissipative channels.  
This is the analytic solution of the equation of motion 
$$\dot\rho_t=L_t\rho_t$$
which is also known as the Gorini–Kossakowski–Lindblad–Sudarshan master equation.
Throughout this work we assume full unitary controllability, i.e. the Hamiltonians generate $\mf{su}(n)$, so that in the absence of dissipation the reachable set of $\rho$ is exactly its unitary orbit.

The purpose of this paper is to understand how the addition of dissipation changes this picture.  
Dissipative terms in the control equation remove spectral invariants and allow motion between distinct unitary orbits, but they also introduce irreversibility and constraints arising from complete positivity.  
From a control theoretic perspective, this leads to a basic question: \emph{how much dissipation is needed to make the dynamics transitive on the full state space $\mc D(\mc H)$?}  
Equivalently, which minimal families of Lindblad generators, when combined with full unitary control, suffice to make every state reachable from every other state?

We address this question in a finite-dimensional Markovian setting.  
Our analysis is structural rather than model specific: we study algebraic, geometric, and dynamical resources and obstructions to controllability, identify simple dissipative resources that break unitality and unitary orbit constraints, and develop a geometric description of admissible state space directions.  
A central theme is the role of Lindbladians as generators of all physically allowable infinitesimal motions on $\mc D(\mc H)$, which allows us to connect global reachability questions with the local geometry of the state space.

\subsection{What is Known}
The modern theory of quantum control rests on the recognition that quantum dynamics, though governed by unitary evolution, can be manipulated in analogy with classical control systems.  
In the Hamiltonian (closed system) setting, the task is to steer the state $\rho(t)$ through a bilinear Schrödinger equation
\begin{equation}
  i\frac{d}{dt}\ket{\psi(t)} = \lpara H_0 + \sum_j u_j(t) H_j\rpara \ket{\psi(t)},
\end{equation}
where the controls $u_j(t)$ modulate external fields.  
Foundational results by Brockett \cite{brockett1972,brockett1977} and by Jurdjevic and Sussmann \cite{jurdjevic1972} established the Lie algebraic rank condition for controllability on compact Lie groups, later systematized in geometric control formalisms \cite{isidori1985,zabczyk2020}.  
This framework underpins modern coherent control theory \cite{dalessandro2021,nielsen2010}, where controllability and accessibility are determined by the Lie algebra generated by the controlled Hamiltonians.  
In this regime, the reachable set of $\rho$ is its unitary orbit, a smooth homogeneous manifold that forms the basis of optimal control, landscape, and feedback approaches \cite{khaneja2001,brif2010,chakrabarti2007,glaser2015,wiseman2009,mirrahimi2004,mirrahimi2005}.

As quantum technologies matured, attention shifted from isolated systems to open quantum systems that interact with their environments.  
Dissipation and decoherence, once viewed as obstacles, are now recognized as resources that can be engineered and exploited, which is the central idea of dissipative or reservoir engineering \cite{verstraete2009,diehl2008,lin2013,mirrahimi2014,leghtas2013,ofek2016,kapit2016}.  
In this setting, dynamics are described by the Lindblad master equation \cite{gorini1976,lindblad1976},
\begin{equation}
  \dot{\rho}(t) = -i[H(t),\rho(t)] 
  + \sum_k \Bigl(L_k(t)\rho(t)L_k^\dagger(t) 
  - \tfrac{1}{2}\lset L_k^\dagger(t)L_k(t), \rho(t)\rset \Bigr),
\end{equation}
which generates completely positive trace-preserving (CPTP) maps on the convex state space $\mc D(\mc H)$ \cite{alicki2007,breuer2002, rivas2012}.  
Here, $H(t)$ and $\{L_k(t)\}$ represent coherent and dissipative control channels, respectively, and controllability becomes a question of semigroup geometry rather than Lie-group structure.

The first rigorous formulation of controllability for open quantum systems was given by Altafini \cite{altafini2003}, who showed that dissipation allows motion between unitary orbits, expanding the reachable set beyond purely coherent limits.  
Subsequent works by Schirmer and collaborators \cite{schirmer2001,dalessandro2010,mirrahimi2007,schirmer2004} and by Grivopoulos and Bamieh \cite{grivopoulos2003} refined these results for Markovian master equations, while Lyapunov and feedback-based methods \cite{mirrahimi2007,mirrahimi2005} highlighted the role of dissipative convergence.  
These studies established that openness qualitatively changes controllability: Hamiltonian invariants such as eigenvalue spectra no longer constrain motion, and reachable sets become convex rather than group-structured.

A general algebraic formulation followed in the Lie semigroup or Lie wedge approach developed by Dirr, Helmke, Schulte-Herbrüggen, Kurniawan, and coauthors \cite{dirr2009,schulte2011,kurniawan2012,altafini2003} and by Schirmer and D'Alessandro \cite{schirmer2004}.  
Here, the generators of Markovian dynamics form a convex cone whose exponential semigroup represents all CPTP evolutions consistent with admissible Lindblad operators.  
Accessibility and controllability are characterized by convexity and semigroup generation rather than Lie closure \cite{lokutsievskiy2021,khaneja2000}, linking geometric control to dissipative engineering and feedback stabilization \cite{ticozzi2008,ticozzi2012}.  
Parallel developments have addressed stabilization \cite{ticozzi2008,ticozzi2012,mirrahimi2014} and invariant subspaces \cite{lidar1998,zanardi1997,viola1999}, while resource theoretic and thermodynamic perspectives \cite{lostaglio2019,brandao2015,baumgratz2014,horodecki2013,ng2015} further clarified the operational meaning of dissipation.

\subsection{Main Results}
Recent developments in quantum information theory have highlighted that dissipation and noise, traditionally regarded as detrimental, can in fact be harnessed to enhance state preparation and control. In parallel, experimental progress has made it possible to engineer and modulate Markovian environments with a high degree of precision. Motivated by these advances, we examine controlled open system dynamics of the general form
$$\dot\rho_t = -i\lbra H_0 + \sum_j u_j(t) H_j, \rho_t\rbra + \sum_k \gamma_k(t)\,L_k(\rho_t),$$
and ask which dissipative mechanisms fundamentally enlarge the set of accessible quantum states. Our focus is on identifying minimal classes of Lindblad generators that, when combined with unrestricted coherent control, allow one to prepare arbitrary density operators in $\mc D(\mc H)$.

Assuming full unitary control, meaning that the set $\{H_0, H_j\}$ forms a Hörmander system\footnote{a subset that generates the Lie algebra via Lie brackets} of $\mathfrak{su}(n)$, our first objective is to determine the minimal dissipative resources required for controllability. We then extend this line of inquiry by examining the interplay between coherent and incoherent controls: how Hamiltonian and dissipative resources can jointly enable reachability within the state space. To approach these questions, we study the geometric structure of $\mc D(\mc H)$ and identify how dissipation reshapes its accessible regions.

In \cref{uni+diss}, we show that remarkably little dissipation is needed to render $\mc D(\mc H)$ controllable; often, as few as two independent jump operators suffice. The Hamiltonian dynamics provides the dominant evolution, while dissipation acts as a subtle yet essential mechanism for breaking unitary constraints and expanding the reachable set.  
In \cref{lift}, we characterized the tangent space on the space of densities, and showed the only allowable directions are given by Lindbladians. We also propose and start to address the physical realizability of trajectories, establishing criteria that distinguish which paths in the space of density operators can be implemented in practice.
In \cref{reach}, we investigate accessibility and controllability properties, including no-go results that identify fundamental limits under constrained resources.  
Finally, in \cref{alg}, we illustrate these ideas through algorithmic examples, demonstrating that controllability can be achieved optimally within specific, physically motivated resource configurations.

\subsection*{Acknowledgment}
The author wishes to acknowledge with huge gratitude the thought provoking conversations with Jake Xuereb, Florian Meier, and Paul Erker at Beyond IID 12 at University of Illinois Urbana-Champaign, which played an important role in inspiring the inception of this project. The author is also grateful to Frederik vom Ende for pointing out the reference \cite{hiriart2004} relevant to \cref{char}, as well as for drawing our attention to their related work \cite{dirr2019}, which contains ideas closely connected to those developed in \cref{qubit_trans}.

\section{Unitary with A Little Help}\label{uni+diss}
Throughout this paper, we assume that the Hilbert space $\mc H$ is finite-dimensional. In fact, many of our results fail to hold in the infinite-dimensional setting, or remains unknown.
Let $\mb B(\mc H)$ denote the algebra of bounded linear operators on $\mc H$. A density matrix $\rho$ is a positive semidefinite operator with unit trace, that is, $\rho \ge 0$ and $\tr(\rho) = 1$.,
The set of all such densities, known as the space of densities, is denoted by 
$$\mc D(\mc H) = \lset \rho \in \mb B(\mc H) : \rho = \rho^*,\, \rho \ge 0,\, \tr(\rho) = 1\rset .$$

Quantum systems are inevitably noisy due to unavoidable interactions with their environments. To describe the dynamics of such open systems, one must go beyond unitary evolution. Under the Markovian assumption, which excludes memory effects, and the semigroup property, the time evolution of an open quantum system is governed by a one-parameter semigroup of completely positive and trace-preserving (CPTP) maps $(T_t)_{t \ge 0}$, generated by a linear operator $L$ known as a Lindbladian. The semigroup is expressed as $T_t = e^{tL}$. According to the Gorini-Kossakowski-Sudarshan-Lindblad theorem \cite{gorini1976,lindblad1976}, any such generator admits the canonical form
\[
L(\rho) = -i[H, \rho] + \sum_{j=1}^n \lpara L_j \rho L_j^* - \tfrac{1}{2}\lpara L_j^*L_j\rho+\rho L_j^*L_j\rpara \rpara,
\]
where $H = H^*$ is the Hamiltonian governing the unitary part of the evolution, and the operators $L_j$ represent the dissipative interactions with the environment. The set of all such generators is denoted by
$$\mc L = \lset L : L \text{ is a Lindbladian}\rset.$$
When we wish to exclude the unitary part $-i[H, \cdot]$ and focus only on the dissipative contributions, we will write $L_{\mathrm{diss}}$ or specify the definition explicitly.

For notational convenience, we write $\ad_H = [H, \cdot]$ for the commutator superoperator and $\Ad_U = U(\cdot)U^*$ for unitary conjugation. Quantum channels, denoted by $\Phi$, are completely positive and trace-preserving maps, and can be viewed as discrete-time analogues of Lindbladian evolutions. In particular, for any channel $\Phi$, the difference $\Phi - \id$ is a Lindbladian, whose generator is given by the Kraus oeprators, where the Kraus operators of $\Phi$ describe the infinitesimal generators of the dissipative dynamics. Hence, a quantum channel may be regarded as a time $t$ snapshot of a continuous Markovian process. The replacer (erasure) channel that maps every input state to a fixed output state $\sigma$ is written as $\mc R_\sigma(\cdot) = \tr(\cdot)\,\sigma$. We also denote by $X, Y, Z$ the Pauli matrices, by $I$ the identity matrix, and by $\mathrm{id}$ the identity superoperator acting between matrix algebras.

These conventions and definitions establish the mathematical framework for our discussion. In the following sections, we explore how even a small amount of dissipation can produce powerful effects in the dynamics of open quantum systems.

\subsection{Minimal Dissipation for Controllability}
In the literature, open-system quantum control is typically studied under the assumption of full unitary controllability, often accompanied by generous amount of dissipative resources. However, as we will demonstrate, even a small amount of dissipation can have a profound impact: with full unitary control, only minimal dissipation is sufficient to ensure controllability over the entire space of densities.
Inspired by \cite{lip_complexity, lip_simulation} we start with a resource set $S$ and investigate the properties of the evolutions $\mc E(S)$ induced by $S$. Define the following properties: 
\begin{enumerate}
    \item[(UN)] If $iH\in \spn\lset S\rset$ is an anti-hermitian operator, then the unitary channel $\Phi_t(\rho)=e^{itH}\rho e^{-itH}$ belongs to $\mc E(S)$ for $t \geq 0$. 
    \item[(JU)] For any operator $a\in S$ we define the Lindblad generator $L_a(\rho)=a^*a\rho+\rho a^*a-2a\rho a^*$ and declare that $ e^{tL_a}\in \mc E(S)$ for $t \geq 0$
    \item[(CO)] The class $\mc E(S)$ is closed under composition.
    \item[(CL)] The class $\mc E(S)$ is closed.
    \item[(CV)] The class $\mc E(S)$ is convex.
\end{enumerate}
Let $\chan(S)$ denote the class of evolutions with property (UN), (JU) and (CO). We write $\overline{\chan(S)}$ to mean if in addition $\chan(S)$ is closed, i.e. with (CL), and denote $\chan^*(S)$ if in addition the class is closed under convex combinations, i.e. with (CV).

It is common to assume the ability to implement arbitrary unitary controls on the system. Such control is a powerful tool: it allows one to move between states on the same unitary orbit and is sufficient for tasks such as vector state steering. For a detailed introduction on unitary control, see \cite{dalessandro2021}. In fact, it already suffices to assume access to a generating set of the special unitary group, i.e. a Hörmander system.

\begin{definition}
    Let $S \subset \mb{B}(H)$ be a set of anti-Hermitian operators on a Hilbert space of dimension $d$. We say that $S$ is a \emph{Hörmander system} if the Lie algebra generated by $S$ spans all of $\mathfrak{su}(d)$, that is,
    $$\mf{su}(d) = \text{span} \lset  [H_{j_1}, [H_{j_2}, \dots [H_{j_{m-1}}, H_{j_m}]]] : H_j \in S \rset.$$
\end{definition}
For a detailed survey on Hörmander system and the geometry of Lie algebra, see \cite{hormander}. When $d<\infty$, it is known that the minimal Hörmander system for $\mf{su}(d)$ is of size 2 \cite{kuranishi1951}, though the two are rather complicated is not available in many practical settings.

Much of the analysis replies on using combinations of the following evolutions in $\chan(S)$.
\begin{proposition}\label[proposition]{chan_membership}
    Let $S$ be a Hörmander system of $\mf{su}(2^k)$. Then 
    \begin{enumerate}
        \item $\Ad_{its}\in \chan(S)$ if $s=s^*$ and $s\in\spn(S)$.
        \item $e^{tK}\in\chan(S)$ where $K=i\ad_H+\sum_j\gamma_jL_{a_j}$ for $a_j\in S$ and $\gamma_j\geq 0$.
        \item $e^{tL_{\Ad_u(a)}}\in\chan(S)$ and also $e^{t\sum_j L_{\Ad_{u_j}({a_j})}}\in\chan(S)$ for unitary $u_j$ with $a_j\in S$.
    \end{enumerate}
\end{proposition}
\begin{proof}
    \begin{enumerate}
        \item Since we are in finite dimension, this is Chow's theorem \cite{chow}  
        \item The Trotter formula 
            $$e^{tL_1+sL_2} = \lim_n (e^{\frac{t}{n}L_1}e^{\frac{s}{n}L_2})^n$$
            holds in Banach algebras \cite{hall2013quantum}. Thanks to (i), we know that $\Ad_{e^{itH}}\in \chan(S)$. By the jump rule (JU), we know that $L_a$ is an admissible jump operator. Iterative application of Trotter's formula implies the assertion. 
        \item It is easy to see that 
            \[ u^*L_a(u\rho u^*)u
            = u^*a^*au\rho + \rho u^*a^*au-2u^*au\rho u^*a^*u 
            = L_{u^*au}(\rho)  . \] 
            By exponentiation 
            \[ e^{tL_{\Ad_{u^*}a}} = \Ad_{u^*}e^{tL_a}\Ad_{u} \] 
            is an element in $\chan(S)$ using (i) and (CO). The additional assertion follows from (ii).
    \end{enumerate}
\end{proof}
It is clear that the operators above are Lindbladians, and therefore generate completely positive semigroups, i.e. Markovian dynamics in the sense of \cite{wolf_cirac}.

\begin{lemma}\label[lemma]{cms} 
\begin{enumerate} 
    \item Let $a= \kla \begin{array}{cc} 0 &1 \\ 
        0& 0\end{array}\mer$ and $\rho_{\lambda} = \kla\begin{array}{cc} \frac{1}{\lambda+1} & 0 \\ 0 &\frac{\lambda}{\lambda+1}\end{array}\mer$. The Lindbladian
        \[ L_a(\beta) = \beta^{1/2}L_a+\beta^{-1/2}L_{a^*}  .\]
        satisfies $L_a(\beta)(\rho_{\lambda})=0$ and $\lim_{t\to \infty} e^{tL(\beta)}(\sigma)=\rho_{\lambda}$ if $1+\lambda= \beta$. 
    \item Let $\rho_{\mu}= {\rm diag}(\mu_1,\dots ,\mu_d)$ be a diagonal state. Let $a_{r} = |r\ran\lan r+1|$ and
        \[ L_{\mu} = \sum_{r} \lpara \beta_r^{1/2} L_{a_r}+\beta_r^{-1/2}L_{a_r^*}\rpara  , \beta_r = \frac{\mu_r}{\mu_{r+1}} . \] 
        Then $L(\rho_{\mu})=0$ and $\lim_{t\to\infty} e^{tL}(\sigma)= \rho_{\mu}$.      
\end{enumerate} 
\end{lemma}

\begin{proof} For (i) The relation between $\lambda$ and $\beta$ is exactly the detailed balance condition of \cite{CM17}. For the convergence we refer to \cite{Jchen}. Note here that $L(\beta)$ satisfies a spectral gap.
    For (ii), we use $a_r=\ket r\bra{r+1}$ and the Lindbladian 
    \[ L = \sum_{r} \beta_r^{1/2}L_{a_r}+\beta_r^{-1/2} L_{a_r^*}  , 
    \beta_r = \frac{\mu_r}{\mu_{r+1}}  .\]
    According to \cite{CM17} we know that 
    \[ L^*(\rho_{\lambda}) = 0  .\] 
    Let us also recall that gradient formula \cite[(2.6)]{JLRR} 
    \[ \langle L(X),X\rangle_{\rho} =  \sum_{r} \langle [a_r,X],[a_r,X]\rangle +
    \langle [a_r^*,X],[a_r^*,X]\rangle  .\] 
    Thus the state of invariant densities $\sigma= \rho_{\mu}^{1/2}X\rho_{\mu}^{1/2}$ satisfies $[a_r,X]=0=[a_r^*,X]$ for all $r$. This means $X=1$ and we have a primitive semigroup. In particular, $L$ has a spectral gap and 
    $\lim_t e^{tL^*}(\sigma)=\rho_{\mu}$. 
\end{proof}

We are interested in understanding the (minimal) set of resources $S$ so that $\mc D(\mc H)$ is controllable. We first clarify two definitions.
\begin{definition}
    The space of densities $\mc D(\mc H)$ is said to be controllable under a set of operations $\mc O$ if for all $\rho,\sigma\in\mc D(\mc H)$, there exists some $\Phi\in\mc O$ such that $\Phi(\rho)=\sigma$. In the continuous setting, for all $\rho,\sigma\in\mc D(\mc H)$, there exists $T_t\in\mc O$ such that $T_\tau(\rho)=\sigma$ for some $\tau>0$ or in the limit as $t\rightarrow\infty$.
\end{definition}
\begin{definition}
    A set of operations $\mc O$ is called transitive if $\mc D(\mc H)$ is controllable under $\mc O$.
\end{definition}
  
\begin{theorem} \label{trans1}
    Let $S=S^*$ be such that ${\rm span}(S)$ is a H\"ormander system and $a=\ket 0\bra 1\in S$. Then $\chan(S)$ is transitive. 
\end{theorem}

\begin{proof}
    Let $\rho_1$ be a state and $\rho_2=\rho_{\mu}$ be diagonal. By approximation we may assume that $\rho_{\mu}$ is faithful.  Recall that $e^{tL_{u^*au}}$ and $e^{tL_{u^*a^*u}}$ in $\chan(S)$. We use the jump operators $a_r$ from \cref{cms} and find unitaries $u_r$ such that $u_r^*\ket 0\bra 1u_r=a_r$. By \cref{chan_membership}, we deduce that 
    \[ e^{tL_{\mu}} \in \chan(S) \] 
    where $L_{\mu}$ is the Lindbladian from \cref{cms}. Then the limit $E_{\mu}=\lim_{t\to \infty} e^{tL_{\mu}}(\rho_1)=\rho_{\mu}$ is the replacer channel for $\rho_{\mu}$.  If $\rho_2$ is not diagonal, we apply this argument first to $\hat{\rho}_2$ given by the singular values. Another unitary rotation yields the assertion and convexity was not needed.  
\end{proof}

In many applications, the jump operators are not rank one. So our condition may be challenging for a given system $S$. We need, however, at least one non-self-adjoint element: 

\begin{proposition}\label[proposition]{sa}
    Let $S$ be a set of self-adjoint elements. Then ${\overline{\chan^*(S)}}$ consists of unital channels. In particular, ${\overline{\chan^*(S)}}$ is not transitive. 
\end{proposition} 
\begin{proof}
    The channels $\Ad_{e^{itH}}$ for self-adjoint $H$  are certainly unital. For $Y=Y^*$, the Lindblad generator
    \[ L(x) = (2YxY-Y^2x+xY^2)\] 
    is unital and self-adjoint with respect to the trace, i.e. $tr(L(x)^*y)=tr(x^*L(y))$. Thus $L^*$ the generator on density matrices is also unital. This $L^*(\frac{1}{d}I)=0$ implies that $e^{tL^*}(\frac{1}{d}I)=\frac{1}{d}I$. This remains true for linaer combination of self-adjoint jump operators.  In other words in finite dimension the maximally mixed state is an eigenvector for all the components $e^{tL}$ and $\Ad_{e^{itH}}$. A fast glance at the closure procedures (CO) and (CL) will preserve this property. Thus the maximally mixed state can not be moved.     
\end{proof}
Note that $\overline{\chan^*(S)}$ is considerably larger than $\chan(S)$, and certainly much larger than the subclasses used throughout our argument as described in \cref{chan_membership}. In fact, $\overline{\chan^*(S)}$ contains many channels that belong only to the class of infinitesimally divisible channels in the sense of \cite{wolf_cirac}, whereas all the channels appearing in the transitivity proof, those in \cref{chan_membership}, are genuine (time independent) Markovian channels, since they are all of Lindblad form. In the transitivity result above, one may replace $\chan(S)$ with the subclass of Markovian channels prescribed in \cref{chan_membership}.

We see that the failure in \cref{sa} can be rescued easily if we allow just one non-self-adjoint resource.

We will use the notation $|0_k\ran=\ket{0\cdots0}$ for the $k$-fold tensor product of $|0\ran$.
\begin{theorem} \label[theorem]{alice}
    Let $\mc H=\mb C^{2^k}$. Let $S$ be a H\"ormander system for $\mf{su}(2^k)$ and
    \[ a = \ket 0\bra 1\otimes I^{\ten{k-1}} \in S  .\]
    Then $\overline{\chan^*(S)}$ contains all the replacer channels, i.e. $\mc R_\sigma\in\overline{\chan^*(S)}$ for all $\sigma\in\mc D(\mc H)$. In particular $\overline{\chan^*(S)}$ is transitive.  
\end{theorem}     

\begin{proof}
    We may consider $a=\ket 0\bra 1\otimes I^{\ten{k-1}}$ and 
    \[L = \sum_{j=1}^n L_{a_j} \] 
    given by moving $a_j$ to the $j$-register, thanks to \cref{chan_membership}. Then $E_{\ket{0_k}\bra{0_k}} = \lim_{t\to \infty} e^{tL}$ is in $\chan(S)$. Using (iii) from \cref{chan_membership} we see that $\Ad_uE_{\ket{0_k}\bra{0_k}}= E_{|u(0_k)\ran\langle u(0_k)|}$ is also in $\chan(S)$. By assumption $\chan^*(S)$ is convex and hence contains every replacer channel. In particular $\chan^*(S)$ is transitive.  
\end{proof} 

Our next aim is to remove the convexity assumption in \cref{alice} by adding Lindbladians. We use standard notation $X,Y,Z$ for Pauli matrices and $V_j$ for the copy in the $j$-th register. We also define
\[ E_{\infty}\kla \begin{array}{cc} \rho_{00}& \rho_{01} \\ 
\rho_{10}& \rho_{11} \end{array}\mer =
\kla \begin{array}{cc} \rho_{00}& 0 \\ 
0 & \rho_{11} \end{array}\mer \] 
for the projection onto the diagonal in one qubit. Moreover, $E_{\infty}^{k-1}=E_{\infty}\otimes \cdots \otimes E_{\infty}\otimes id$ is the tensor product of $k-1$ such projections.    

Further, we define the following classes
\begin{definition}
    For a given density $\rho$, the \emph{range} is defined
    \[ D_S(\rho) = \lset \Phi(\rho) :\Phi\in \overline{\chan(S)}\rset \qquad D_S^*(\rho) = \lset \Phi(\rho) :\Phi\in \overline{\chan^*(S)}\rset  .\]
\end{definition} 

\begin{remark}
    Since we are in finite dimension $D_S^*(\rho)$ is convex. At the time of this writing the condition on $S$ for the convexity of $D_S(\rho)$ remain elusive.
\end{remark}

\begin{lemma}\label[lemma]{addZ}
    Suppose $\mc H=\mb C^{2^k}$ is the Hilbert space of $k$ qubits. Let $S=S^*$ be generate a H\"ormander system. 
    \begin{enumerate}
        \item If $\ket 0\bra 1\otimes I^{k-1}\in S$, then 
        \[ \lset \rho_1\otimes \cdots \otimes \rho_{k}: \rho_j\in \mc D(\mb C^2)\rset  \subset  D_S(\ket{0_k}\bra{0_k})  .\] 
        \item If  $Z_1\in S$, then the conditional expectation $E_{\infty}\otimes \id^{k-1}$ and $E_{\infty}^{\ten{k}}$ are in $\overline{\chan(S)}$.
        \item[(iii)] If $\ket 0\bra I\otimes 1^{k-1}$ and $Z_1$ are in $S$, then 
        \[ D_{S}(\ket{0_k}\bra{0_k}) = \mc D(\mc H) \]
        is controllable.
    \end{enumerate}  
\end{lemma}

\begin{proof}
    Let us recall \cite{CW} that for $a=\ket 0\bra 1$ we have
    \[ e^{tL_a}\kla \begin{array}{cc} \rho_{00}& \rho_{01} \\ 
    \rho_{10}& \rho_{11} \end{array}\mer =
    \kla \begin{array}{cc} \rho_{00}+(1-e^{-2t})\rho_{11}& e^{-t}\rho_{01} \\ 
    e^{-t}\rho_{10} & e^{-2t}\rho_{11} \end{array}\mer  .\] 
    We see that by varying $e^{tL_a}(\ket 1\bra 1)$ we get all positive diagonal densities.  The same is true for $e^{tL_{a^*}}(\ket 0\bra 0)$. Thus for a tensor product $\rho=\rho_{1}\otimes \cdots \otimes \rho_{k}$ we may just chose $e^{t_1L_{a^*}}\otimes \cdots \otimes e^{t_kl_{a^*}}$ appropriately. Note, however, that this is composition of the channels $e^{t_jL_{a_j^*}}= \Ad_{u_j} e^{tL_{a_1^*}}\Ad_{u_j}$ where $u_j$ is the tensor flip between the first and the $j$-th register. Thus \cref{chan_membership} allows us to produce the same Lindbladian in the $j$-th registers. Using the H\"ormander system once more we can use conjugation by $w=w_1\otimes \cdots \otimes w_k$ to produce every tensor product that to the standard singular value decomposition.

    For the proof of (ii), we just have to note that $\lim_{t\to \infty} e^{tL_Z}=E_{\infty}$ because 
    \[ L_Z(\rho) = 2Z\rho Z-\rho  = 4\lpara\frac{\rho+Z\rho Z}{2}-\rho\rpara
    = 4(E_{\infty}-id)(\rho)  .\] 
    Using again the conjugation trick and products we see that for every subset $A$ the tensor product $E_{\infty}^{A}$ in $A$-register belongs to $\overline{\chan(S)}$.

    For the proof of (iii) we start with any diagonal $(f(\omega))_{\omega \in \lset 0,1\rset^k}$. Then we can find a unitary such that $u(|0_k\ran)= (\sqrt{f(\omega)})_{\omega}$ is the unit vector given square root. We deduce 
    \[ {\rm diag}(f) = E_{\infty}^k(u\ket{0_k}\bra{0_k}u^*)  .\] 
    For an arbitrary density we first produce the density ${\rm diag}(f)$ given by the eigenvalues and the add a rotation to create $\rho$.
\end{proof}

\begin{theorem} \label{qubit_trans}
    Let $S=S^*$ generate a H\"ormander system for $\mc H=\mb C^{2^k}$ and $\ket 0\bra 1\otimes I^{\ten{(k-1)}}$ and $Z\otimes I^{\ten{(k-1)}}\in S$. Then $\overline{\chan(S)}$ is transitive.  
\end{theorem} 

\begin{proof}
    Let $\mc R_{\ket 0\bra 0}=\lim_{t\to \infty}e^{tL_a}$ be the replacer channel. By assumption, $\mc R_{\ket 0\bra 0}\otimes \id$ and  the tensor product $\mc R_{\ket{0_k}\bra{0_k}}$ are in $\overline{\chan(S)}$. Let $\rho,\sigma$ be densities. We first destroy all information $\mc R_{\ket{0_k}\bra{0_k}}(\rho)=\ket{0_k}\bra{0_k}$ and then recreate $\sigma$ thanks to \cref{addZ}.  
\end{proof}
It appears that the strategy of first destroying information to prepare a reference pure state, followed by controlled transport to the desired target state, is standard in this area (e.g. \cite{dong2008, dirr2019}). What remains unclear, however, is the explicit role played by \cref{addZ} in this context.

\subsection{Linearity of Resources and Bilinear Gradient Form}\label{obstructions}
A natural question in any resource-based formulation is how a given set of generators behaves under composition. 
In the unitary setting, access to Hamiltonians $H_1$ and $H_2$ automatically implies access to their sum $H_1+H_2$, reflecting both the linear structure of the Schr\"odinger equation and the standard Lie-algebraic viewpoint on controllability \cite{JurdjevicSussmann1972,dalessandro2021}. 
This intuition also underlies much of the resource theoretic literature, where free sets are typically assumed to have convexity or closure properties \cite{ChitambarGour2019,BrandaoGour2015}. 
In the dissipative context, however, it is no longer clear whether such closure should be expected.

Note that in our previous definition of channel class, it is not assumed. 
We only allow $L_a+L_b$, for $a,b\in S$, to be two independent dissipative processes, each coupled to their uncorrelated baths, but disallow the collective, correlated dissipative channel $L_{a+b}$, producing correlated decay and interference effects.
The obstruction for this analysis is the non-linearity of the map $a\mapsto L_a$, see \cite{childs2016} for a similar discussion of sparse Lindbladians. Indeed, a better description is to look at a sesquilinear map 
 \[ (a,b) \mapsto L_{a,b}(\rho) := 2a^*\rho b-a^*b\rho-\rho a^*b  .\]
Using diagonalization it is easily follows that for a positive definite matrix $\gamma_{j,k}$ we obtain a Lindbladian 
 \[ L = \sum_{j,k} \gamma_{j,k} L_{a_j,a_k} \]
for all choices $a_j$. The Lindbladian corresponding to the sum is a prominent example 
\begin{equation}\label{sum}
 L_{a+b} = L_{a,a} + L_{b,b}+L_{a,b}+ L_{b,a} 
\end{equation}
given by the $2\time 2$ matrix with all entries $1$.
Our remedy is to consider Lindblads gradient form 
 \[ \Gamma_{L}(x,y) = L(x^*y)-L(x)^*y-x^*L(y)  \] 
The map $L\to \Gamma_{L}$ is not injective because it forgets derivations. A derivation is a linear map such that 
 \[ \delta(xy) = x\delta(y)+\delta(x)y  .\] 
Such a derivation is called self-adjoint (preserving) if $\delta(x^*)=\delta(x)^*$. The map $\delta(x)=i[H,x]$ for self-adjoint $H$ is an example of a self-adjoint derivation. The following result is well-known:

\begin{lemma}\label[lemma]{jj}
    Let $\delta$ be a self-adjoint derivation and $L$ be a Lindbladian. Then
    \[ \Gamma_{\delta+L} = \Gamma_{L}.\] 
    In particular 
    \[ \Gamma_{L_{a+\lambda 1}} = \Gamma_{L_a}.\]  
\end{lemma}   

The first assertion follows from linearity and $\Gamma_{\delta}=0$. For the second assertion we note that $L_c$ is always self-adjoint preserving. Then \eqref{sum} 
shows that for $b=\lambda 1$, the difference $L_{a+\lambda 1}-L_{a}$ is a derivation. Inspired by \cite{group_gradient_flow}, the following definition extends the class of Lindbladians associated with a resource set. 

\begin{definition}
    Let $S=S^*$ be a resource set. Then
    \[ L_{\Gamma}(S) = \lset L  : \exists{L'\in L(S)}  \text{ where } \Gamma_{L'} - \Gamma_{L} \text{ is completely positive} \rset   .\] 
    Then $\chan_{\Gamma}(S)$ is the set of channels closed under the operations (CO) and (CL) containing all $e^{tL}$ with $L\in L_{\Gamma}(S)$. Here convexity is not required.
\end{definition} 
The following result from \cite{JW} clarifies this definition. 

\begin{theorem}[\cite{JW}]
    Let $L=\sum_j L_{a_j}$ and $L'=\sum_k L_{b_k}$ be two Lindbladians. The following are equivalent
    \begin{enumerate}
        \item There exists a constant $c>0$ with 
            \[ \Gamma_{L}(x_j,x_l) \kl c \Gamma_{L'}(x_j,x_j) \]
            holds all $n$ tuples $(x_1,..,x_n)$. 
        \item $a_j$ belongs to the span of $\lset 1,b_1,\dots ,b_m\rset $.
    \end{enumerate} 
\end{theorem}

\begin{corollary}
    The set $\lset a| L_a\in L_{\Gamma}(S)\rset $ is a linear space.  
\end{corollary} 

\begin{lemma}\label[lemma]{spann}
    Let $a^*\neq az$ for all $z\in \mb C$. Then the linear span of $ \lset u^*au|u\in u(d)\rset \cup\lset u^*a^*u|u\in u(d)\rset $ contains $b=\ket 0\bra 1$. 
\end{lemma} 

\begin{proof} 
    Let us consider the conjugate representation $\pi(u)(x)=u^*xu$. The real invariant subspaces of $\ell_2^n\otimes \overline{\ell}_2^n=S_2^n$ are 
    \[ K_1= \mb R 1 , K_{sym} = \lset x: x^*=x\rset   , K_{ant} = \lset x:x^*=-x\rset   \] 
    by \cite{approx_cloning}. For any $a$, we may consider the real invariant subspace $$K_a=\lset\sum_j \lambda_j u_j^*au_j: \lambda_j\in \mb R, u_j\in U(n)\rset.$$ If $K_a\cap K_j\neq 0$, then $K_j\subset K_a$. Let us assume that $tr(a)=0$. Note that $a'=a-\frac{tr(a)}{d}\in K_a$. For a non-self-adjoint $a$, we know that $K_a$ is not contained in $K_{sym}$ and $K_a\cap K_1=\lset 0\rset $. Thus $K_a$ has to contain an element in $K_{ant}$. Thus $K_{ant}\subset K_a$. However, $iK_{ant}=K_{sym}$. This $K_{sym}+iK_{sym}$ is contained in the complex orbit $K_a^{\mb C}$. Since $b=\ket 0\bra 1$ has trace $0$, we obtain the assertion.
\end{proof}

\begin{theorem}
    Let $S=S^*$ be a H\"ormander system such that $S\subset L(A)$ contains an element which is neither self-adjoint or antisymmetric. Then $\chan_{\Gamma}(S)$ is the set of all channels and 
    \[ \lset L_a: a\in L(A)\rset  \subset L_{\Gamma}(S)    .\] 
\end{theorem}  

\begin{proof}
    Let us consider
    \[ A(S) = \lset a : L_a\in L_{\Gamma}(S)\rset   .\]  
    Since $S$ is H\"ormander, we know that $L_{\Gamma}(S)$ is invariant under $\Ad_u$ conjugation and contains $\ket 1\bra 2$ and $\ket 2\bra 1$, hence all matrix units, hence all of $L(A)$ by linearity.  Thus $L_{\Gamma}(S)$ is the set of all channels.
\end{proof}

\begin{remark}
    It would be nice to have a more operational description for the passage from $L(S)$ to $L_{\Gamma}(S)$.
\end{remark}

\subsection{Environment-Assisted Controllability}
We will conclude our investigation by adding environment to our resource set 
 \[  S_{AE} = S_A\otimes 1\cup S_A\otimes X\cup  1\otimes S_E\qquad S_E=\lset X,Y\rset   .\] 
Since we want to implement a class of channels, we have to add state preparation 
 \[ E_{prep}(\rho_A\otimes \rho_E)= tr(\rho_E) \rho_{A}\otimes \ket 0\bra 0 \] 
to the set of allowable operations
\begin{enumerate}
    \item[(UN)] If $iH\in {\rm span}(S_{AE})$ is anti-Hermitian, then $\Phi_t(\rho)=e^{itH}\rho e^{-itH}$ belongs to $\chan_{AE}(S)$ for $t\in \mb R$. 
    \item[(PR)] $E_{prep}\in \chan_{AE}(S)$
    \item[(CO)] If we generate several channels $\Phi_{t_1},\dots ,\Phi_{t_m}\in \chan_{AE}(S)$, then the composition is also generated $\Phi_{t_1}\cdots\Phi_{t_m}\in \chan_{AE}(S)$;
    \item[(CL)] If $ \Phi \in \overline{\chan_{AE}(S)}$ then $ \Phi \in \chan_{AE}(S)$; that is to say $\chan_{AE}(S)$ closed set. 
    \item[(CC)] $\chan_{AE}(S)$ is convex.      
\end{enumerate}
 
 \begin{definition}
    $\chan_A(S)$ is the set of channels of the form
    \[ \Phi(\rho) = tr_E(\Psi(\rho\otimes \ket 0\bra 0)) \] 
    where $\Psi\in \chan_{AE}(S)$.
\end{definition}

 \begin{theorem}\label{htrans}
    If the span of $S_A$ contains a H\"ormander system, then $e^{tL}\in \chan_A(S)$ for every Lindbladian $L$ on $A$. 
 \end{theorem}

\begin{lemma}
    If $S_A$ is H\"ormander, then $S_{AE}$ is H\"ormander.
\end{lemma} 

\begin{proof}
    Any operator $H\in L(AE)$ can be written as
    \[ H = H_1\otimes 1+ H_2\otimes X+H_2\otimes Y+H_3\otimes Z  \]
    Thus for $H=H^*$, we may replace $H_j$ by $\frac{H_j+H_j^*}{2}$. Since $S_A$ is H\"ormander, we can find $iH_j$ in span of the iteration commutators 
    \[ iH_j = \sum_{k_1,\dots ,k_m} \alpha(k_1,\dots ,k_m) [s_{k_1},[\cdots ,s_{k_m}]] \] 
    with $s_k\in {\rm span} S$, $s_k^*=-s_k$. Replacing the last component by $s_{k_m}\otimes iY$ we find 
    \[ iH_j \otimes iY  .\] 
    Using the commutator relation of the Pauli matrices, we find $iH$ in the span and $u_t=e^{itH}$ in the unitary group thanks to Chow's theorem \cite{chow}
\end{proof}
\begin{remark}
    Let $\Phi$ be a set of channels on $B$ and $|E'|$ a one qubit environment.  The convex combination 
    \[ \frac{1}{2}\Phi_1+\frac{1}{2}\Phi_2 = tr_E \kla \begin{array}{cc} \Phi_1&0 \\
    0&\Phi_2 \end{array}\mer (id\otimes X_{E'}) {\rm prep}_{0}  \] 
    initialization channel ${\rm prep}_0(\rho)=\rho\otimes \ket 0\bra 0$, and the direct sum channel. In our situation, for two  Hamiltonians $H_1,H_2$ on $AE$ and $S_{E'}=\lset X,Y,Z\rset $, we can prepare
    \[ H = H_1\otimes \frac{1+Z_{E'}}{2} + H_2 \otimes \frac{1-Z_{E'}}{2} \] 
    Then 
    \[ \Ad_{e^{itH}} =   \kla \begin{array}{cc} \Ad_{e^{itH_1}}&0 \\
    0&e^{itH_2}\end{array}\mer  .\] 
    Therefore, we can avoid the convexity assumption $(CC)$ by adding two bits $EE'$ of environment preparation and partial trace out channel.  
\end{remark}  

\begin{lemma}\label[lemma]{eps}
    Let $H$ be self-adjoint (and bounded). Then 
    \[ \norm{e^{tL_H} - \frac{\Ad_{e^{i\sqrt{2t}H}}+\Ad_{e^{-i\sqrt{2t}H}}}{2} }_\diamond = O(t^2)     . \] 
\end{lemma}

\begin{proof}
    Let us write $H=\sum_j \lambda_j e_j$ with eigenvalues $\lambda_j$ and eigen-projections $e_j$. Then
    \begin{align*}
        \Ad_{e^{itH}}(x) &=  \sum_{j,l} e^{it(\lambda_j-\lambda_l)}e_jxe_l   , \\
        \Ad_{e^{-itH}}(x) &=  \sum_{j,l} e^{it(\lambda_l-\lambda_j)}e_jxe_l   , \\
        e^{tL_H}(x) &= \sum_{j,l} e^{t(\lambda_j-\lambda_l)^2} e_jxe_l  .
    \end{align*}
    The assertion follows Since $\cos(\sqrt{2t}\lambda)=1-t\lambda^2+ O(t^2)$.      
\end{proof}

Let us recall a result from \cite{lip_simulation} 
\begin{lemma}
    Let $a\in L(A)$ be an operator and
    \[ H = \kla \begin{array}{cc} 0& a\\ a^* & 0 \end{array}\mer . \] 
    Then 
    \[  tr_EL_H{\rm prep}_0 = L_a  .\] 
    If $iH$ is in the Lie algebra generated by $S_{AE}$, then 
    \[  e^{tL_a} \in \chan_{A}(S) .\]   
\end{lemma} 

\begin{proof}
    We just observe that 
    \[ L_H(\rho\otimes \ket 0\bra 0) = \kla \begin{array}{cc} 
    -a^*a\rho-\rho a^*a & 0 \\
    0 & 2 a\rho a^*\end{array} \mer  .\]
    Taking the trace in $E$ gives the first assertion. Thus
    \[ \|tr_Ee^{tL_H}{\rm prep}_0-e^{tL_a}\|_{\diamond} \kl Ct^2  .\] 
    We deduce from \cref{eps} that 
    \[ \norm{tr_E \lpara\frac{\Ad_{e^{i\sqrt{2t}H}}+\Ad_{e^{-i\sqrt{2t}H}}}{2} -e^{tL_a}\rpara
    {\rm prep}_0}_{\diamond} \kl C' t^2  . \]
    Let us denote the $\Psi_t$ the first channel (obtained by convexification or adding $E'$). Now, we can use Trotterization
    \[ e^{tL} = \lim_n (e^{t/nL})^n = \lim_n (\Psi_{t/n})^n  .\] 
    For convergence see \cite{banach_algebras} applied to the space of channels as a Banach algebra with the diamond norm.    
\end{proof}

\begin{proof}[Proof of \cref{htrans}]
    Since $S_{AE}$ gives rise to a H\"ormander system on the combined space, we deduce that  $e^{tL_a}\in \chan_A(S)$ for all $a$. Using Trotterization and the fact that $S_A$ induces a  H\"ormander we can generate all channels $e^{tL}$ for all Lindbladian  $L$.
\end{proof}
 
We leave two questions open for future investigation.
First, we want to consider the case of a resource set that does not form a Hörmander system, and ask whether transitivity and controllability can still be achieved. If so, what is the minimal set of dissipative resources required to ensure these properties?
Second, we ask whether it is possible to achieve transitivity using a small resource set acting on the joint system–environment. In this setting, resources are allowed to act on the combined system, and the question is whether controllability of the reduced system can be obtained after taking the partial trace over the environment.

\section{Lift me up}\label{lift}
Since not every path on $\mc D(\mc H)$ is physical -- for example it can fail complete positivity -- a fundamental question in open quantum control is whether every desired infinitesimal change of a quantum state $\dot{\rho}_t$ can be generated by a physically valid (completely positive and trace-preserving) dynamical law, and whether this correspondence can be made smoothly in time. In other words, can one design a smooth feedback law that maps the instantaneous state $\rho_t$ and its desired rate of change $\dot{\rho}_t$ to an admissible open-system generator $L_t$? The regularity of this correspondence -- its continuity, differentiability, or analyticity in time -- reflects the degree of controllability and stability of open quantum dynamics. A continuous dependence ensures physically realizable and stable control laws, while smoothness enables differential feedback, optimization, and geometric analysis. Analyticity, in turn, supports perturbative and adiabatic constructions where the generator varies without singularities.

Mathematically, this question can be formulated as a \emph{lifting problem} for the time-dependent evaluation map
\[
\ev_t: \mc D(\mc{H}) \times \mc L \to T^+\mc D(\mc{H}), 
\qquad 
(\rho_t, L_t) \mapsto (\rho_t, L_t(\rho_t)),
\]
which associates to each instantaneous generator its induced velocity on the manifold of density matrices. The existence of a continuous or smooth retract of this map provides precisely such a rule: assigning to each admissible trajectory $\rho_t$ a corresponding family of generators $L_t$ that realize it, thus offering a geometric criterion for the physical realizability of infinitesimal state motions.

That is, we seek a family of maps
\[
s_t : T^+\mc D(\mc{H}) \to \mc D(\mc{H}) \times \mc L
\]
such that $\ev_t \circ s_t = \mathrm{id}_{T^+\mc D(\mc{H})}$ for all $t\in [0,\infty)$, and where the assignment 
$t \mapsto s_t((\rho_t, \dot{\rho}_t)) = (\rho_t, \mc{L}_t)$ 
is continuous, smooth, or analytic in time.

\subsection{Lindbladian Tangent Cone}
The first question to ask is why we believe Lindbladian is the correct object to lift? Physically, it is natural to conjecture that all physical evolutions should be infinitesimally satisfied by the solution of Markovian master equation. The following theorem shows this physical intuition is valid mathematically.
\begin{definition}
    Given a state $\rho\in\mc D(\mc{H})$, its tangent vector (allowable velocities) is given by the derivative of any $C^2$ paths. 
    The collection of all tangent vectors is the tangent cone, given by 
    $$T_\rho^+\mc D(\mc H)=\lset \left.\dfrac{d}{dt}\right|_{t=0}\rho_t: \rho_0=\rho, \rho_t\in C^2([0,\infty), \mc D(\mc H))\rset$$
    The tangent bundle is defined as 
    $$T^+\mc D(\mc H)=\coprod_{\rho\in\mc D(\mc H)}T^+_\rho\mc D(\mc H).$$
\end{definition}
Even $\mc D(\mc H)$ is convex and compact, its tangent cone is not the convex tangent cone. There are more allowable directions than the convex ones on boundary points. Here is an example:
\begin{proposition}  \label{nonconv} For non-invertible state $\rho \in \partial \mc D(\mc{H})$
    $$T^+_\rho \mc D(\mc{H})\supsetneq \lset x\in\mb B(H): x=x^*, \tr x=0, \rho+\varepsilon x\geq 0\rset.$$
\end{proposition}
\begin{proof}
    The inclusion is immediate since if $\rho + \varepsilon x \geq 0$ for some $\varepsilon > 0$, then the curve $\rho + \varepsilon x$ lies in the state space $\mc D(\mc{H})$ for small $\varepsilon$, so $x \in T^+_\rho \mc D(\mc{H})$.
    
    To show that equality does not hold, we construct an example where a Hermitian traceless operator $x$ lies in the tangent cone, but $\rho + \varepsilon x \not\geq 0$ for any $\varepsilon > 0$.
    Let
    $$\rho = \begin{pmatrix} 1 & 0 \\ 0 & 0 \end{pmatrix}, \quad
    x = \begin{pmatrix} 0 & 1 \\ 1 & 0 \end{pmatrix}.$$
    Then
    $$\rho + \varepsilon x = \begin{pmatrix} 1 & \varepsilon \\ \varepsilon & 0 \end{pmatrix}$$
    has determinant $\det(\rho + \varepsilon x) = -\varepsilon^2 < 0$, so it is not positive semidefinite for any $\varepsilon > 0$. Hence, $x$ is not in the set on the right-hand side.
    
    However, define
    $$b = \begin{pmatrix} -2 & 0 \\ 0 & 2 \end{pmatrix},$$
    and consider the second-order perturbation:
    $$\rho_\varepsilon = \rho + \varepsilon x + \varepsilon^2 b = 
    \begin{pmatrix}
    1 - 2\varepsilon^2 & \varepsilon \\
    \varepsilon & 2\varepsilon^2
    \end{pmatrix}.$$
    The determinant of this matrix is
    $$\det(\rho_\varepsilon) = (1 - 2\varepsilon^2)(2\varepsilon^2) - \varepsilon^2 = 2\varepsilon^2 - 4\varepsilon^4 - \varepsilon^2 = \varepsilon^2(1 - 4\varepsilon^2).$$
    For small $\varepsilon > 0$, this is positive, and the matrix is clearly Hermitian with trace $1$. Hence, $\rho_\varepsilon \in \mc D(\mc{H})$ for small $\varepsilon$, and the curve $\rho_\varepsilon$ lies in the state space. Its derivative at $\varepsilon = 0$ is $x$, so $x \in T^+_\rho \mc D(\mc{H})$.
    
    Therefore, $x \in T^+_\rho \mc D(\mc{H})$, but $x$ is not in the set $\lset  x : \rho + \varepsilon x \geq 0 \rset $, showing that the inclusion is strict.
\end{proof}

Even though the state space is convex, we see that the tangent cone at a corner point is not given by linear paths, we must consider all arbitrary $C^2$ paths if we want to understand the tangent space. this also shows that analytic considerations, whether we consider $C^2$ paths or $C^n$ paths, may possibly affect the tangent cone at a corner point. A possible direction for future research is weakening the analytic condition to $C^1$. It turns out that the $C^2$ condition is enough for the next characterization, which we will then connect to the time-irreversibility of quantum Markovian semigroup. 

\begin{theorem}[\cite{CM20, hiriart2004}]\label{char}
    Let $\rho\in \mc D(\mc{H})$ and $f:\mc H\rightarrow\supp\rho$ be the projection onto the support of $\rho$, then 
    $$T_\rho^+\mc D(\mc{H}) = \lset x\in\mb B(H):x=x^*, \tr x=0, (1-f)x(1-f)\geq 0\rset\\$$
\end{theorem}
To prove that \cref{char} works for all states in $\mc D(\mc{H})$ we first need the following diagonalization lemma.
\begin{lemma}[\cite{paulsen}]\label[lemma]{diag_trick}
    Let $A$ be an invertible matrix. We write as a block matrix:
    $$\begin{pmatrix}
        A & B\\
        B^* & C 
    \end{pmatrix}\geq 0\Longleftrightarrow  A \geq 0, -B^*A\inv B+C\geq 0$$
\end{lemma}
    
Using this lemma, we can now show that our characterization holds for all states. 
\begin{proof}[Proof of \cref{char}]
    Let $\rho_t \in C^2([0,\infty),\mc D(\mc{H}))$, for small $t$, we can write it as 
    $$\rho_t=\rho+tx+ t^2b(t)$$ for some continuous $b(t)$ that is time-dependent. Its derivative at $0$ is $\dot \rho_0 =x$.

    We know $x$ is Hermitian and traceless since $\rho_t\in \mc D(\mc{H})$ for all $t$. We write as a block matrix
    $$\rho=\begin{pmatrix}
        \rho_{11} & 0\\
        0 & 0
    \end{pmatrix}\quad x=\begin{pmatrix}
        x_{11} & x_{12}\\
        x_{21} & x_{22}
    \end{pmatrix}.$$
    We need to show that $x_{22}\geq 0$. We know that 
    $$\rho_t=\begin{pmatrix}
        \rho_{11}+tx_{11}+t^2b_{11}(t) & tx_{12}+t^2b_{12}(t)\\
        tx_{21}+t^2b_{21}(t) & tx_{22}+t^2b_{22}(t)
    \end{pmatrix}.$$
    By \cref{diag_trick}, $\rho_t\geq 0$ if and only if
    \begin{itemize}
        \item $\rho_{11}+tx_{11}+t^2b_{11}(t)\geq 0$ and 
        \item $-(tx_{21}+t^2b_{21}(t))(\rho_{11}+tx_{11}+t^2b_{11}(t))\inv(tx_{12}+t^2b_{12}(t))+(tx_{22}+t^2b_{22}(t))\geq 0$
    \end{itemize}
    Since $\rho_{11}$ is invertible, $\rho_{11}+tx_{11}+t^2b_{11}(t)$ is also invertible for small enough $t$ and the inverse is analytic by von Neumann series. Thus, by expanding out the second term, we get that $tx_{22}+t^2b_{22}(t)\geq 0$ for all $t$ where the expansion of the path holds. This is satisfied only if $x_{22}\geq 0$ with arbitrary higher order terms.

    To show the reverse inclusion, we need to show that if $x$ satisfies the condition, we can construct a path whose first order term is $x$.
    If $x$ is invertible, then there always exists a small $\varepsilon$ such that $\rho+tx\geq 0$ for all $0<t<\varepsilon$. When $x$ is non-invertible,
    write as a block  matrix $$\rho=\begin{pmatrix}
        \rho_{11} & 0 & 0\\
        0 & 0 & 0\\
        0 & 0 & 0
    \end{pmatrix}\quad x=\begin{pmatrix}
        x_{11} & x_{12} & x_{13}\\
        x_{21} & x_{22} & x_{23}\\
        x_{32} & x_{32} & 0
    \end{pmatrix}.$$
    First, to ensure $(1-f)x(1-f)\geq 0$, it is necessary that $x_{23}=x_{32}=0$.

    In this case, we claim there exists an $x_2$ such that $\rho+tx+t^2x_2\in \mc D(\mc{H})$ for all $t<\varepsilon$. In particular, 
    $$x_2=\begin{pmatrix}
        0 & 0 & 0\\
        0 & 0 & 0\\
        0 & 0 & B
    \end{pmatrix}$$ suffices. That means, we need to show that there exists $B$ so that
    $$\rho_t=\begin{pmatrix}
        \rho_{11}+tx_{11} & tx_{12} & tx_{13}\\
        tx_{21} & tx_{22} & 0\\
        tx_{31} & 0 & t^2B
    \end{pmatrix}\geq 0$$
    for small enough time $t<\varepsilon$.

    To show this, we will apply \cref{diag_trick} several times. $\rho_t \geq 0$ if and only if 
    $$\begin{pmatrix}
           \rho_{11}+tx_{11}-t^2\frac{\tr B}{n} & tx_{12}\\
           tx_{21} & tx_{22} 
        \end{pmatrix}\geq 0\,\text{and}\,
        -\begin{pmatrix}
            tx_{31} & 0
        \end{pmatrix}\begin{pmatrix}
            y_{11} & y_{12}\\
            y_{21} & y_{22}
        \end{pmatrix}\begin{pmatrix}
            tx_{13}\\0
        \end{pmatrix}+\begin{pmatrix}
            t^2B
        \end{pmatrix}\geq 0$$
    where $$\begin{pmatrix}
        y_{11} & y_{12}\\
        y_{21} & y_{22}
    \end{pmatrix}=\begin{pmatrix}
        \rho_{11}+tx_{11}-t^2\frac{\tr B}{n} & tx_{12}\\
        tx_{21} & tx_{22} 
    \end{pmatrix}\inv$$ which is time dependent. The first one holds by \cref{diag_trick}, since 
    \begin{itemize}
        \item $\rho_{11}+tx_{11}-t^2\frac{\tr B}{n}\geq 0$
        \item $-t^2x_{21}\lpara \rho_{11}+tx_{11}-t^2\frac{\tr B}{n}\rpara\inv x_{12}+tx_{22}\geq 0$
    \end{itemize}
    are true for small time $t$. To check the second one, 
    $$-\begin{pmatrix}
        tx_{21} & 0
    \end{pmatrix}\begin{pmatrix}
        y_{11} & y_{12}\\
        y_{21} & y_{22}
    \end{pmatrix}\begin{pmatrix}
        tx_{13}\\0
    \end{pmatrix}+\begin{pmatrix}
        t^2B
    \end{pmatrix}=-\begin{pmatrix}
        t^2x_{31}y_{11}x_{13}
    \end{pmatrix}+\begin{pmatrix}
        t^2B
    \end{pmatrix}.$$
    This is positive if and only if $$B\geq x_{31}y_{11}x_{13}.$$
    But notice that $y_{11}$ is time-dependent, and we want a time-independent choice of $B$. That means, we need to check that $y_{11}$ does not perturb too much as time goes.
    It suffices to find a time-independent upper bound for $y_{11}$ to find a time-independent $B$ making the second matrix positive.
    Notice that 
    \begin{equation*}
    \begin{aligned}
        &\quad \begin{pmatrix}
            \rho_{11}+tx_{11}-t^2\frac{\tr B}{n} & tx_{12}\\
            tx_{21} & tx_{22} 
        \end{pmatrix}\begin{pmatrix}
            y_{11} & y_{12}\\
            y_{21} & y_{22}
        \end{pmatrix}\\
        &=\begin{pmatrix}
            \lpara \rho_{11}+tx_{11}-t^2\frac{\tr B}{n}\rpara y_{11}+tx_{12}y_{21} & \lpara \rho_{11}+tx_{11}-t^2\frac{\tr B}{n}\rpara y_{12}+tx_{12}y_{22}\\
            tx_{21}y_{11}+tx_{22}y_{21} & tx_{21}y_{12}+tx_{22}y_{22}
        \end{pmatrix}\\
        &=\begin{pmatrix}
            1 & 0\\
            0 & 1
        \end{pmatrix}.
    \end{aligned}
    \end{equation*}
    Thus, $tx_{21}y_{11}+tx_{22}y_{21}=0$ implies $y_{21}=-x_{22}\inv x_{21}y_{11}$
    and 
    \begin{equation*}
    \begin{aligned}
        1 &= \lpara \rho_{11}+tx_{11}-t^2\frac{\tr B}{n}\rpara y_{11}+tx_{12}y_{21}\\
        &= \lpara \rho_{11}+tx_{11}-t^2\frac{\tr B}{n}\rpara y_{11}+tx_{12}\lpara -x_{22}\inv x_{21}y_{11}\rpara\\
        &= \lpara \rho_{11}+tx_{11}-t^2\frac{\tr B}{n}-tx_{12}x_{22}\inv x_{21}\rpara y_{11}
    \end{aligned}
    \end{equation*}
    Hence, since $$\frac{1}{2}\rho_{11}\leq \rho_{11}+tx_{11}-t^2\frac{\tr B}{n}-tx_{12}x_{22}\inv x_{21}\leq 2\rho_{11}$$ for small $t$.
    So $y_{11}= \lpara \rho_{11}+tx_{11}-t^2\frac{\tr B}{n}-tx_{12}x_{22}\inv x_{21}\rpara\inv\leq 2\rho_{11}\inv$. Thus, 
    $$B\geq 2x_{31}\rho_{11}x_{13}\geq x_{31}y_{11}x_{13}.$$
    Thus, we have checked that both operators are positive semi-definite. 
\end{proof}

This characterization allows us to show an additional fact: the tangent cone at $\rho$ is precisely the set of observables achievable by applying a Lindbladian to a state $\rho$
\begin{theorem}\label{lind_tangent}
    For $\rho\in \mc D(\mc{H})$,
    $$T_\rho^+\mc D(\mc{H})=\lset L(\rho):L\text{ is Lindbladian}\rset.$$
\end{theorem}
Recall that Lindbladians are the infinitesimal generators of quantum Markovian semigroup. It turns out that that the time-irreversibility of a quantum Markovian semigroup precisely lines up with the tangent cone conditions at the boundary of $\mc D(\mc{H})$. We first show the result for interior states where the boundary conditions are not at play:

\begin{proof}
    We first show this characterization for invertible densities and will use graph Lindbladians to show it for the general case.

    Given that $(e^{tL})_{t\geq 0}$ forms a quantum Markovian semigroup and its derivative at $t=0$ is $L$, it follows by definition that $L(\rho)$ lies in the tangent cone $T_\rho^+ \mc D(\mc{H})$.

    Fix a tangent vector $x\in T_\rho^+\mc D(\mc{H})$; we need to find a Lindbladian $L$ such that $L(\rho)=x$.
    Consider the replacer channel We know $\rho+\varepsilon x$ is a state for a small enough $\varepsilon>0$ since $\rho$ is an invertible state. This allows us to define the replacer channel $\Phi(\eta)=(\tr\eta)(\rho+\varepsilon x)$. 
    Then $L=\frac{1}{\varepsilon}(\Phi-\id)$ is a Lindbladian, and $$L(\rho)=\frac{1}{\varepsilon}(\Phi-I)(\rho)=\frac{1}{\varepsilon}(\rho+\varepsilon x-\rho)=x.$$
    
    Since $(e^{tL})_{t\geq 0}$ is a quantum Markovian semigroup whose derivative at $t=0$ is $L$. It is clear that $L(\rho)\in T_\rho^+\mc D(\mc{H})$ by definition.

    To show the reverse inclusion, we need to find a Lindbladian for any $x\in T_\rho^+\mc D(\mc{H})$.
    Without loss of generality, assume that both $\rho$ and $x$ are reduced to the blocks by their supports $$\rho=\begin{pmatrix}
        \rho_{11} & 0 & 0\\
        0 & 0 & 0\\
        0 & 0 & 0
    \end{pmatrix}\quad x=\begin{pmatrix}
        x_{11} & x_{12} & 0\\
        x_{21} & x_{22} & 0\\
        0 & 0 & 0
    \end{pmatrix}$$

    We claim there exists some $a, b_j$ such that $$L_a(\rho)+\sum L_{b_j}(\rho)=x$$
    where $$a=\begin{pmatrix}
        0 & 0 & 0\\
        a_{21} & a_{22} & 0\\
        0 & 0 & 0
    \end{pmatrix} \quad b_j=\begin{pmatrix}
        b_{j,11} & 0 & 0\\
        0 & 0 & 0\\
        0 & 0 & 0
    \end{pmatrix}$$
    By direct calculation, 
    $$L_a(\rho)=\begin{pmatrix}
        -\rho_{11}a_{21}^*a_{21}-a_{21}^*a_{21}\rho_{11} & -\rho_{11}a_{21}^*a_{22} & 0\\
        a_{22}^*a_{21}\rho_{11} & 2a_{21}\rho_{11}a_{21}^* & 0\\
        0 & 0 & 0
    \end{pmatrix}\quad L_{b_j}=\begin{pmatrix}
        L_{b_{j,11}}(\rho_{11}) & 0 & 0\\
        0 & 0 & 0\\
        0 & 0 & 0
    \end{pmatrix}$$
    Then $L_a(\rho)+\sum L_{b_j}(\rho)=x$ implies we need to solve the system of equations 
    $$\begin{cases}
        x_{11} = \sum_j L_{b_{j,11}}(\rho_{11}) -\rho_{11}a_{21}^*a_{21}-a_{21}^*a_{21}\rho_{11}\\
        x_{12} = -\rho_{11}a_{21}^*a_{22}\\
        x_{21} = a_{22}^*a_{21}\rho_{11}\\
        x_{22} = 2a_{21}\rho_{11}a_{21}^*
    \end{cases}$$
    Note that both $\rho_{11}$ and $x_{22}$ are positive and full rank by assumption, so the fourth equation becomes
    \begin{equation*}
    \begin{aligned}
        x_{22} &= 2a_{21}\rho_{11}a_{21}^*\\
        \lpara \frac{1}{2}x_{22}\rpara^{\frac{1}{2}} &= (a_{21}\rho_{11}^{\frac{1}{2}})(a_{21}\rho_{11}^{\frac{1}{2}})^*
    \end{aligned}
    \end{equation*}
    Let $a_{21}$ be self-adjoint, and thus 
    $$a_{21}=\frac{1}{\sqrt{2}}x_{22}\rho_{11}^{-\frac{1}{2}}.$$
    In particular, $a_{21}$ is invertible.
    Notice that the second and third equations are duals of each other, and hence 
    $$a_{22}=-a_{21}\inv\rho_{11}\inv x_{12}$$ solves the two equations for any given $x_{12}$.
    Finally, notice that $x_{11}+\rho_{11}a_{21}^*a_{21}+a_{21}^*a_{21}\rho_{11}$ is Hermitian and traceless. This is because
    \begin{equation*}
    \begin{aligned}
        \tr x_{11} &= \tr\lpara\sum_j L_{b_{j,11}}(\rho_{11}) -\rho_{11}a_{21}^*a_{21}-a_{21}^*a_{21}\rho_{11}\rpara\\
        &= \tr \lpara-\rho_{11}a_{21}^*a_{21}-a_{21}^*a_{21}\rho_{11}\rpara
    \end{aligned}
    \end{equation*}
    Hence, by Lindbladian characterization for invertible density, $b_j$ exists making 
    $$\sum_j L_{b_j}=x_{11}+\rho_{11}a_{21}^*a_{21}+a_{21}^*a_{21}\rho_{11}.$$
    Hence, we find the Lindbladian corresponding to any tangent vector.
\end{proof}

We see here that Lindbladians exactly map boundary points to elements in the tangent cone. The time irreversibility of the quantum Markovian semigroup generated by the Lindbladian corresponds to the fact that we infinitesimally cannot leave the state space $\mc D(\mc{H})$. Additionally from the above proof, we see that it is enough to consider only the completely dissipative part of the Lindbladian to fill out the tangent cone. In the next section we consider time and position dependent Lindbladians and the evolutions they generate.

\subsection{Lindbladian Lifting Problem}
We now turn to the question of which paths can be physically realized, meaning those generated by some time-dependent Markovian evolution. In quantum control theory, a open problem is whether the reach, i.e. the set of attainable states given certain resources, coincides across different models of open-system dynamics. In particular, it remains unclear whether the reach of non-Markovian evolutions matches that of Markovian ones under comparable resource constraints \cite{koch2022}.
Here, we pose a stronger question: not only should the reachable sets coincide, but the entire path itself should be physically realizable: every point along the trajectory should correspond to an actual Markovian evolution.
Formally, given a path of density operators $\rho_t$, we can associate a family of tangent pairs $(\rho_t, \dot\rho_t)$. The question is whether there exists a corresponding family of Lindbladians $(\rho_t, L_t)$ such that
$$(\rho_t, L_t(\rho_t)) = (\rho_t, \dot\rho_t).$$
In mathematical terms, this asks whether the evaluation map admits a retract.
\begin{center}
\begin{tikzcd}
    \mc D(\mc H)\times \mc L \arrow[r, "\ev_t"'] & T^+\mc D(\mc{H}) \arrow[l, bend right, dashed, "s_t"']
\end{tikzcd}
\end{center}
We view this as the beginning of a broader line of investigation and take the first steps by formulating sufficient conditions for lifting.

\begin{proposition}[spectral lifting]\label{speclift} Let $\lambda(t)$ be the spectral gap of $\rho_t$. Assume $\rho_t, \dot \rho_t, {\dot \rho_t}^{-1} \in L^{\infty}(I,\mc D(\mc H))$ and $\dot \rho_t \in T^+_{\rho_t}(\mc D(\mc H))$ almost everywhere. If  $\lambda_{min}(t)^{-1}, \sqrt{\lambda_{min}(t)^{-1}} \in L^1(I, \mc D(\mc H))$ there exists a Lindbladian lifting $L_t(\rho_t) = \dot \rho_t$ with $L_t \in L^1(I, \mc D(\mc H))$
\end{proposition}
\begin{proof}
We know that from \cref{char}, we can pick a Lindbladian, $L_t$ that splits into three parts: $L_t = L_{a_{21}(t)} +  L_{a_{22}(t)} + L_{\Phi_t,t}$. We should check that each of these parts are measurable.

We look at the jump operator for the first term, $L_{a_{21}(t)}$

$$a_{21}(t)=\frac{1}{\sqrt{2}}x_{22}(t)\rho_{11}(t)^{-\frac{1}{2}}$$

The integrability of $a_{21}$ is dependent on $\rho_{11}^{-1/2}$. In particular, if $\int_I \sqrt{\lambda_{min}(t)^{-1}} dt < \infty$ we have that $\rho_{11}(t)^{-1/2} \in L^1(I, \mc D(\mc H))$. Since $x_{22}(t) \in L^\infty(I, \mc D(\mc H)$ we have that $L_{a_{21}}(t) \in L^1(I, \mc D(\mc H))$. 

Now we look at the second term:
$$a_{22}(t)=-a_{21}(t)\inv\rho_{11}(t)\inv x_{12}(t) = -\sqrt{2} \rho_{11}(t)^{1/2} x_{22}(t)^{-1} \rho_{11}(t)^{-1}x_{12}(t)^{-1}$$

 It suffices to check $\rho_{11}(t)^{-1} \in L^1(I,\mc D(\mc H))$, since the rest of the terms are in $L^\infty$ by assumption. Indeed, since $\int_I \lambda_{11}(t)^{-1} dt < \infty$ we have that $\rho_{11}(t)^{-1} \in L^1(I, \mc D(\mc H)$ and $L_{a_{22}(t)} \in L^1(I, \mc D(\mc H))$.

Now, let's look at the third term: \begin{align*}L_{\Phi_t,t}(\eta) &= \frac{1}{\lambda_{min}(t)}[\tr(\eta)\rho_t + \tr(\eta) \lambda_{min}(t) - \id] \\ &= \frac{\tr(\eta)\rho_t}{\lambda_{min}(t)} + \tr(\eta) -  \frac{\id}{\lambda_{min}(t)} \end{align*}

 Since $\rho_t \in L^\infty(I, \mc D(\mc H))$ and $\int_0^{\infty} \lambda_{min}(t)^{-1} dt < \infty$ we have $L_{\Phi_t,t} \in L^{1}(I, \mc D(\mc H))$. 

We see that $L_t$ is integrable since $L_{a_{21}(t)}$,  $L_{a_{22}(t)}$, and $L_{\Phi_t,t}$ are individually integrable. 
\end{proof}

Note that if $I$ is a finite interval, $\lambda_{min}(t)^{-1} \in L^1(I, \mc D(\mc H))$ implies $\sqrt{\lambda_{min}(t)^{-1}} \in L^1(I, \mc D(\mc H))$ which slightly weakens the conditions of \cref{speclift}.

\section{Reach me}\label{reach}
Our strategy for studying reachability in restricted sets of Lindbladians is to first start with some time-dependent Lindbladian evolution admissible to some subset of Lindbladians $\mc K$, see that it generates a Lindbladian vector field $L_\eta$ in $\mc K$, and then use that vector field to find local improvements in norm. 

\begin{definition}
We say that an evolution $T_t$ is \emph{admissible} to a subset of Lindbladians $\mc K$ if and only if there exists a family $L_t \in \mc K$ for all $t$ such that $T_t$ is a weak solution to the equation of motion $\dot T_t = L_t T_t$
\end{definition}

\begin{definition}
    We say that a subset of Lindbladians $\mc K$ \emph{reaches $\sigma$} if for every initial state $\rho$, there exists some time-dependent Lindbladian evolution $T_t$ admissible to $\mc K$ with $\lim_{t \to \infty} T_t(\rho) = \sigma$. If $\mc K$ reaches every state $\sigma \in \mc D(\mc{H})$, then we call $\mc K$ \emph{transitive}.
\end{definition}

The following result about transitivity of all Lindbladians is well known.
\begin{theorem}
    The set $\mc L$ of all Lindbladians is transitive. Furthermore, all of the necessary evolutions can be done in finite time. 
\end{theorem}
\begin{proof}
    The controllability of $\mc L$ is well-known. The result is due to the theory of Lie semigroup wedge (see for example \cite{dirr2009, hilgert1989}). For finite time, we consider two cases: interior points and boundary points.
    
    For invertible densities, finite time controllability can be achieved by using the replacer channel path with an overshoot. Define $\mc R_\sigma$ as the replacer channel to the state $\tilde{\sigma} = \sigma + \varepsilon(\sigma - \rho)$, which is a valid density operator as long as $\sigma$ is invertible and $\varepsilon > 0$ is sufficiently small. Then the evolution is given by
    $$T_t(\rho) = e^{-t} \rho + (1 - e^{-t}) \tilde{\sigma}.$$
    Substituting $\tilde{\sigma} = (1 + \varepsilon)\sigma - \varepsilon \rho$, we get
    \begin{align*}
    T_t(\rho)   = ((1 + \varepsilon)\sigma - \varepsilon \rho) + e^{-t} ((1 + \varepsilon)(\rho - \sigma)).
    \end{align*}
    To reach $\sigma$, solve $T_s(\rho) = \sigma$, which yields
    $e^{-s} = \frac{\varepsilon}{1 + \varepsilon}.$ That means, $$ s = \ln\left(1 + \frac{1}{\varepsilon}\right).$$
    Since $\varepsilon > 0$, this transport time is finite.

    For non-invertible densities, we will need to use time-dependent Lindbladian by a rescaling argument. $\Phi_t(\eta) = e^{-t} \eta + (1-e^{-t}) \sigma$. Now using $\dot \Phi_t(\eta) = e^{-t}(\sigma-\eta)$ we introduce a time-compression function $f(t)$. 
    \begin{align*}
    \frac{d}{dt} \Phi_{f(t)}(\eta) = f'(t) e^{-f(t)}(\sigma - \eta)
    \end{align*}
    Now letting $f(t) = \tan(t)$ we get 
    \begin{align*}
    \frac{d}{dt} \Phi_{f(t)}(\eta) &= \frac{1}{\cos(t)^2} e^{-\tan(t)}(\sigma - \eta)
    \end{align*}
    We now see that $\lim_{t \to \pi/2} \frac{d}{dt} \Phi_{f(t)}(\eta) = 0 \in T^+_\eta(\mc D(\mc{H}))$ so $T_t = \Phi_{f(t)}$ is Lindbladian. Moreover, $T_{\pi/2}(\rho) = \sigma$.
\end{proof}

\begin{theorem} \label{transtheorem}
    Fix a $p$-norm with $1 < p < \infty$ and a subset of Lindbladians $\mc K$. If for $\sigma \in \mc D(\mc{H})$ we can find some Lindbladian vector field $L_\eta \in \mc K$  so that for all $\eta \neq \sigma \in  \mc D(\mc{H})$ we have $$\tr(L_\eta(\eta)(\eta - \sigma)\abs{\eta -\sigma}^{p-2}) < 0$$ then $\mc K$ reaches $\sigma$. 
\end{theorem}
\begin{proof}
    We first calculate the derivative from \cite{Bhatia}
    \begin{align*}
    \frac{d}{dt} \frac{1}{p} \norm{\rho_t - \sigma}^p_p = \tr(\dot \rho_t (\rho_t - \sigma) \abs{\rho_t - \sigma}^{p-2})
    \end{align*}
    If the condition of the theorem holds, for any $\rho$ we can always select a $T_t$ admissible to $\mc K$ with $\dot T_t(\rho) = L_{T_t(\rho)}(\rho)$ for all $t \in [0,\infty]$ so that $\frac{d}{dt} \norm{ T_t(\rho) - \sigma}^p_p < 0$ almost everywhere. Note that this $T_t$ may have countably many pieces for each choice of Lindbladian.
    We see that $\liminf_{t \to \infty} \norm{ T_t(\rho) - \sigma}^p_p = c$ exists since $\inf \norm{ T_t(\rho) - \sigma}^p_p$ is monotonic and bounded below by $0$.

    Now for contradiction, assume $c \neq 0$. We have that there is some $s \in [0,\infty]$ so that $\norm{ T_s(\rho) - \sigma}^p_p = c$ for all $t \geq s$ and $\frac{d}{dt}\Bigr|_{t=s} \norm{T_t(\rho) - \sigma}^p_p = 0$. Letting $\eta = T_s(\rho)$ we see that $\tr(L_\eta(\eta)(\eta - \sigma)\abs{\eta -\sigma}^{p-2}) = 0$ which contradicts the assumption.

    This implies $\liminf_{t \to \infty} \norm{ T_t(\rho) - \sigma}^p_p = 0$ which in turn means that $T_\infty(\rho) = \sigma$.
\end{proof}
\begin{center}
\begin{figure}[t!]
    \begin{tikzpicture}[scale=1.5, line cap=round, line join=round, >=Triangle]
    \coordinate (O) at (0,0);
    \fill (O) circle (0.8pt);
    \node[below] at (O) {$\sigma$};

    \coordinate (S) at (-1.5, 1);
    \fill (S) circle (0.8pt);
    \node[below] at (S) {$\eta$};
    \draw[->,red] (-1.5, 1) -- (0,0);
    \node[below, red] at (-1, 0.5) {$\sigma-\eta$};
    \draw[->, blue] (-1.5, 1) -- (-0.3,1.2);
    \node[above, blue] at (-0.9, 1.1) {$L_\eta$};
    \end{tikzpicture}
    \caption{Illustration of the alignment condition of \cref{p2condition}}
    \label{fig:alignment}
\end{figure}
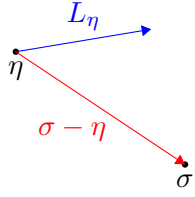
\end{center}

By setting $p=2$ this condition takes on an especially nice form:
\begin{corollary}\label[corollary]{p2condition}
    Fix a subset of Lindbladians $\mc K$. If for $\sigma \in \mc D(\mc{H})$ we can find some Lindbladian vector field $L_\eta \in \mc K$ so that for all $\eta \neq \sigma \in  \mc D(\mc{H})$ we have $$\langle L_\eta(\eta),\sigma-\eta \rangle_{HS} > 0$$ then $\mc K$ reaches $\sigma$. 
\end{corollary}
\begin{proof}
    Recall the Hilbert-Schmidt norm $\langle x,y \rangle = x^* y$ and use that a Lindbladian applied to any state results in a self-adjoint operator. Note here that $\langle L_\eta(\eta),\sigma-\eta \rangle_{HS}$ can be thought of as representing the angle between the tangent vector $L_\eta(\eta)$ and the vector pointing towards our desired final state $\sigma - \eta$. If these two vectors are aligned we can always ensure reachability to $\sigma$. A simple illustration of this alignment can be seen in \cref{fig:alignment}.
\end{proof}

We can use the same Lindbladian vector field idea as \cref{transtheorem} to find topological obstructions to reachability. 

\begin{theorem}[Porcupine Theorem]\label{porcupine}
    Fix some $p$-norm $1 < p < \infty$ and a subset of Lindbladians $\mc K$. Fix a state $\sigma \in \mc D(\mc{H})$ and let $B_{\sigma,\varepsilon} = \lset\eta \in \mc D(\mc{H}): \norm{\sigma-\eta}_p < \varepsilon  \rset$, be the $p$-ball of radius $\varepsilon$ around $\sigma$. Now, if there exists $ \varepsilon > 0$ with the ball strictly contained in the state space $B_{\sigma,\varepsilon} \subset \mc D(\mc{H})$ so that for all $ \eta \in \partial B_{\sigma,\varepsilon} \cap \mc D(\mc{H})$ and for all $L \in \mc K$ we have $$\tr(L(\eta)(\eta-\sigma)\abs{\eta-\sigma}^{p-2}) \geq 0$$ then $\mc K$ does not reach $\sigma$.
\end{theorem}
\begin{proof}
    Assume for contradiction that $\mc K$ is transitive. We have some ball $B_{\sigma,\varepsilon}$ that satisfies the above condition. Then we can find some$\rho \in \mc D(\mc{H})$ with $\rho \notin B_{\sigma,\varepsilon}$. Since $\mc K$ is assumed to be transitive, we can find a time-dependent Lindbladian evolution $T_t$ admissible to $\mc K$ with $\dot T_t(x) = L_t(T_t(x))$ so that $T_s(\rho) = \sigma$ for $s \in [0,\infty]$. This implies that $\norm{T_s(\rho)-\sigma}_p^p = 0$. Meanwhile, we know that $\norm{T_0(\rho) - \sigma}_p^p > \varepsilon^p$. Therefore, there must be at least one point $t_\varepsilon \in [0,s]$ with the property that $\norm{T_{t_\varepsilon}(\rho) - \sigma}_p^p = \varepsilon^p$ and $\frac{d}{dt} \norm{T_{t_\varepsilon}(\rho) - \sigma}_p^p \leq 0$. We calculate $$\frac{d}{dt} \frac{1}{p} \norm{T_{t_\varepsilon}(\rho) - \sigma}_p^p  = \tr(\dot T_{t_\varepsilon}(\rho) (T_{t_\varepsilon}(\rho) - \sigma) \abs{T_{t_\varepsilon}(\rho)  - \sigma}^{p-2}) \leq 0 $$
    Now we know that $L_{t_\varepsilon} = \dot T_{t_\varepsilon}(T_{t_\varepsilon}) \in \mc K$ since $T_t$ is admissible to $\mc K$. Letting $\eta = T_{t_\varepsilon}(\rho)$ we have that $\eta \in \partial B_{\sigma,\varepsilon}$ and $\tr(L_{t_\varepsilon}(\eta)(\eta-\sigma)\abs{\eta-\sigma}^{p-2}) < 0$ which contradicts what was taken. 
\end{proof}
We call the above theorem the Porcupine Theorem, which essentially consists of finding a vectors which create a zone of avoidance for state preparation. 

\begin{example}\label{Noise}
    We take the space $H = \mb C^3$ with $a_1 = \ket 0 \bra 1$, and $a_2 = \ket 1 \bra 2$ being the lowering operators. We take $L_a(\rho) = - a^* a \rho - \rho a^* a + 2 a \rho a^*$ and calculate for a diagonal matrix $\rho =\lambda_0 \ket 0 \bra 0 + \lambda_1 \ket 1 \bra 1 + \lambda_2 \ket 2 \bra 2 $, $$L_{a_1}(\rho) = - \lambda_1 \ket 0 \bra 0 + \lambda_1 \ket 1 \bra 1 \quad L_{a_2}(\rho) = -\lambda_2 \ket 1 \bra 1 + \lambda_2 \ket 2 \bra 2$$   We see here that $L_{a_i}$ preserves the commutative space of diagonal states.  Now we take a state with a significant component in $\ket 0 \bra 0$ and take, \begin{align*} \eta := (1 - 2 \varepsilon) \ket 0  \bra 0 &  + \varepsilon  \ket 1 \bra 1 + \varepsilon \ket 2 \bra 2 \\ L_{a_1}(\eta) = -\varepsilon \ket 0 \bra 0 + \ket 1 \bra 1 & \quad L_{a_2}(\eta) = -\varepsilon \ket 1 \bra 1 + \varepsilon \ket 2 \bra 2 \end{align*} so we see that $\langle L(\eta), (\ket 0 \bra 0 - \eta) \rangle_{HS} \leq 0$. Following \cref{porcupine_example} we take a small ball around $\ket 0 \bra 0$ and recalling that $L_{a_i}$ preserves diagonals, we see that $\mc K = \lset L_{a_1}, L_{a_2} \rset $ cannot reach $\ket 0 \bra0$.

    Note that taking the raising operators $\lset L_{a^*_1}, L_{a^*_2}\rset $ allows reaching the state $\ket 0 \bra 0$, but only taking the raising operators gives us our zone of avoidance for the Porcupine Theorem. Note that this picture changes drastically if we add in Hamiltonians, as we can use Hamiltonian flow to move  from $ \ket 0  \bra 0$ to $\ket i \bra i$. 
\end{example}
\begin{center}
\begin{figure}[t!]
\begin{tikzpicture}[scale=2.5, line cap=round, line join=round, >=Triangle]
    \coordinate (A) at (-1, 0);
    \node[below] at (A) {$\ket 0\bra 0$};
    \coordinate (B) at (0.6, 0);
    \coordinate (C) at (-0.2, 1.2);
    \draw[thick, red, fill=red, fill opacity=0.4] (-0.6, 0) -- (-0.8, 0.3) -- (A) -- cycle;
    \draw[thick] (A) -- (B) -- (C) -- cycle;

    \foreach \p in {(-1, 0), (-0.6, 0), (-0.2, 0), (0.2,0)} {
        \draw[->] \p -- ++(0:0.25);
        \draw[->] \p -- ++(55:0.25);
    }

    \foreach \p in {(-0.8, 0.3), (-0.4, 0.3), (0, 0.3)} {
        \draw[->] \p -- ++(0:0.25);
        \draw[->] \p -- ++(55:0.25);
    }

    \foreach \p in {(-0.6, 0.6), (-0.2, 0.6)} {
        \draw[->] \p -- ++(0:0.25);
        \draw[->] \p -- ++(55:0.25);
    }

    \foreach \p in {(-0.4, 0.9)} {
        \draw[->] \p -- ++(0:0.25);
        \draw[->] \p -- ++(55:0.25);
    }

    \fill[red] (A) circle (0.8pt);
\end{tikzpicture}
\caption{A choice of raising Lindbladians on the commutative state space of $\mb C^3$.}
\label{porcupine_example}
\end{figure}
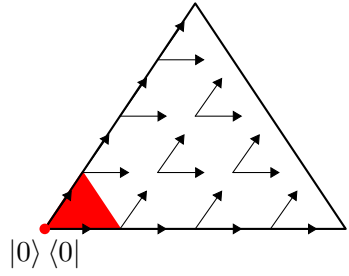
\end{center}

\section{Algorithms and examples}\label{alg}
One interesting application is to consider controllability using a sparse set of Lindbladians.
\begin{example}[Sparse Lindbladian Resources]
    Consider the set of sparse Lindbladians \cite{childs2016} $$\mc K = \lset L_{e_{rs}}\rset \cup \mf{su}(n),$$ where $e_{rs}=\ket r\bra s$ is the matrix unit and $L_{e_{rs}}(\rho) = 2 e_{rs} \rho e^*_{rs} - \rho e_{rs} e^*_{rs} - e_{rs} e^*_{rs} \rho$
    If $\rho$ and $\sigma$ lie in the same unitary orbit, that is, there exists a unitary $u$ such that $\rho = u\sigma u^*$, then $u = \exp(iH)$ for some $H\in\mf{su}(n)$.
    
    Now, suppose $\rho$ and $\sigma$ lie in different orbits. Then we can compute
    \begin{equation*}
    \begin{aligned}
        \tr\lbra L_{e_{rs}}(\rho)(\rho - \sigma) \rbra
        &= 2\rho_{ss}(\rho_{rr} - \sigma_{rr}) - \sum_\ell \lbra\rho_{s\ell}(\rho_{\ell s} - \sigma_{\ell s}) + \rho_{\ell s}(\rho_{s\ell} - \sigma_{s\ell})\rbra \\
        &= 2\rho_{ss}(\rho_{rr} - \rho_{ss} - \sigma_{rr} + \sigma_{ss}) - \sum_{\ell \ne s} \lbra\rho_{s\ell}(\rho_{\ell s} - \sigma_{\ell s}) + \rho_{\ell s}(\rho_{s\ell} - \sigma_{s\ell})\rbra.
    \end{aligned}
    \end{equation*}

    Assume now that both $\rho$ and $\sigma$ are diagonal; then $\rho_{s\ell} = 0$ for all $\ell \ne s$, and the expression simplifies to:
    $$\tr\lbra L_{e_{rs}}(\rho)(\rho - \sigma) \rbra = 2\rho_{ss}(\rho_{rr} - \rho_{ss} - \sigma_{rr} + \sigma_{ss}).$$
    In this case, there always exists a pair $(r, s)$ such that
    $$\tr\lbra L_{e_{rs}}(\rho)(\rho - \sigma) \rbra < 0.$$
    
    To see this, suppose that for all $r, s$,
    $$\rho_{rr} - \rho_{ss} - \sigma_{rr} + \sigma_{ss} \geq 0.$$
    This implies
    $$\rho_{rr} - \sigma_{rr} \geq \rho_{ss} - \sigma_{ss}$$
    for all $r,s$
    which can only hold if $\rho_{rr} - \sigma_{rr} = \rho_{ss} - \sigma_{ss}$ for all $r, s$. Since $\tr \rho = \tr \sigma = 1$, it follows that:
    $$\rho_{rr} = \sigma_{rr}$$
    for all $r$.
\end{example}

The above calculation tells us that the set of sparse Lindbladians is transitive. We note here that the set of sparse Lindbladians is generated by another set $S$, and the transitivity result is intimately connected to the geometry of $S$. This will motivate us to study the transitivity of generating sets more broadly in the next section.

\begin{example}[Transposition with Amplitude-Damping]
    The following model is considered in \cite{bergholm2016}. We presents an algorithmic proof to transitivity under amplitude damping and transpositions. We further discuss the gate count and optimality when similar idea is adapted to different variants of the resource sets.
    
    The goal is to show controllability of the system
    $$\dot\rho_t=i\lbra H_0+\sum_{j}u_j(t)H_j,\rho_t\rbra+v(t)L(\rho_t)$$
    where $u_j(t)\in \mf{su}(2^k)$ and $v(t)\in\lset 0,1\rset$
    with unitary control, i.e. $\lset H_0, H_j\rset$ forms a Hörmander system of $\mf{su}(2^k)$ and the dissipation is generated by the jump operator $\begin{pmatrix}
        0 & 1\\
        0 & 0
    \end{pmatrix}\otimes I\otimes\dots\otimes I$.
    The corresponding semigroup on a single qubit is of the form
    $$e^{-tL_a}=\lbra\begin{array}{cc}
        \rho_{11} & \rho_{12}\\
        \rho_{21} & \rho_{22}
    \end{array}\rbra=\lbra\begin{array}{cc}
        \rho_{11}+(1-e^{-2t})\rho_{22} & e^{-t}\rho_{12}\\
        e^{-t}\rho_{21} & e^{-2t}\rho_{22}
    \end{array}\rbra.$$
    
    Due to unitary control, the question reduced to transport any diagonal density matrix $\rho\sim\diag(\lambda_1,\dots, \lambda_{2^k})$ to any other diagonal density $\sigma=\diag(\mu_1,\dots, \mu_{2^k})$, where $\sum_{j=i}^{2^k}\lambda_j=\sum_{j=i}^{2^k}\mu_j=1$. The transport plan to consider is as follows:
    $$\rho\leadsto\diag(\lambda_1,\dots, \lambda_{2^k})\leadsto \diag(1,0,\dots, 0)\leadsto\diag(\mu_1,\dots, \mu_{2^k})\leadsto \sigma.$$
    It is worth noting that the two parts of this plan are relatively independent. We may therefore consider them as two separate phases: state preparation and state transportation. This means that, regardless of the input diagonal density, we first prepare it into the pure state, and then we show that one can transport from the pure state to an arbitrary diagonal target state.
    The resource set is defined as follows: after the initial diagonalization and the final undo-diagonalization, we are only allowed to use sparse transpositions of the form $X_{j, j+i}=I_{j-1}\oplus X\oplus I_{2^k-j-1}$, where $X$ is the Pauli $X$-matrix on the $j$ and $j+1$ entry in the sparse sense.

    We focus on the portion 
    $$\diag(\lambda_1,\dots, \lambda_{2^k})\leadsto \diag(1,0,\dots, 0)\leadsto\diag(\mu_1,\dots, \mu_{2^k})$$
    and assume we can diagonalize the intial and target states.

    The state preparation portion of the transport plan $\diag(\lambda_1,\dots, \lambda_{2^k})\leadsto \diag(1,0,\dots, 0)$ requires successive amplitude damping that kills half of the mass every time. In particular, we have that 
    \begin{center}
    \begin{tikzcd}
        \diag(\lambda_1,\dots, \lambda_{2^k}) \arrow[d, "AD"]\\
        \diag(\lambda_1+\lambda_{2^{k-1}+1},\dots, \lambda_{2^k-1}+\lambda_{2^k}, 0,\dots, 0) \arrow[d, "TP"]\\
        \diag(\lambda_1+\lambda_{2^{k-1}+1}, \dots, \lambda_{2^{k-2}}+\lambda_{2^{k-1}+2^{k-2}}, 0, \dots, 0, \lambda_{2^{k-2}+1}+\lambda_{2^{k-1}+2^{k-2}+1}, \dots, \lambda_{2^{k-1}}+\lambda_{2^k}, 0,\dots, 0)\arrow[d, "AD"]\\
        \diag(\lambda_1+\lambda_{2^{k-1}+1}+\lambda_{2^{k-2}+1}+\lambda_{2^{k-1}+2^{k-2}+1}, \dots, \lambda_{2^{k-2}}+\lambda_{2^{k-1}+2^{k-2}}+\lambda_{2^{k-1}}+\lambda_{2^k}, 0, \dots, 0)\arrow[d]\\
        \vdots \arrow[d]\\ 
        \diag(1,0,\dots, 0)
    \end{tikzcd}
    \end{center}
    In total, this requires $O(k)$ many infinite time amplitude damping $e^{tL_a}$ and $O(2^{k})$ many sparse transpositions $X_{j, j+1}$.

    For the second part of the transport plan $\diag(1,0,\dots, 0)\leadsto\diag(\mu_1,\dots, \mu_{2^k})$. We first define the ratio between each corresponding pairs to be $$R_m=\frac{d_{2^{k-1}+m}}{d_{m}+d_{2^{k-1}+m}}.$$ The key for the induction is that our algorithm preserves the order of the ratio throughout the process. Thus, we can always obtain a diagonal such that $R_1\leq\dots\leq R_j\leq R_{2^{k-1}}$. Then it is some permutation of length no more than $2^k$ to arrange the diagonal to the desired one.
    
    Then the goal is to find an algorithm to transport pure state to rearrange target state.
    The algorithm is defined recursively:
    \begin{center}
    \begin{tikzcd}
        \diag(1,0,0,0) \arrow[d, "AD"]\\
        \diag(\alpha, 0, (1-\alpha), 0) \arrow[d, "TP"]\\
        \diag(\alpha, 0, 0, (1-\alpha)) \arrow[d, "AD"]\\
        \diag(\alpha, (1-\alpha)(1-\beta), 0, (1-\alpha)\beta) \arrow[d, "AD"]\\
        \diag(\alpha\gamma, (1-\alpha)(1-\beta)\gamma, \alpha(1-\gamma), (1-\alpha)\beta+(1-\alpha)(1-\beta)(1-\gamma))\\
    \end{tikzcd}
    \end{center}
    It is clear that 
    $$\diag(\alpha\gamma, (1-\alpha)(1-\beta)\gamma, \alpha(1-\gamma), (1-\alpha)\beta+(1-\alpha)(1-\beta)(1-\gamma))=\diag (\mu_1, \mu_2, \mu_3, \mu_4)$$
    has a solution, namely
    $$\alpha=\mu_1+\mu_3\quad \beta=\frac{\mu_2+\mu_4}{\mu_1+\mu_3-1}+1\quad \gamma=\frac{\mu_1}{\mu_1+\mu_3}.$$
    In addition, via direct verification, we have that 
    $$R_1=1-\gamma\quad R_2=1-\gamma+\beta\gamma$$
    so $R_1\leq R_2$ checks out.
    
    Now, we apply this procedure inductively. Suppose that we have done it for the first $n-1$ pairs, i.e. 
    we have $$\diag(1, 0,\dots, 0)\leadsto \diag(\lambda_1, \dots, \lambda_{n-1}, 0,\dots, 0, \lambda_{2^{k-1}+1}, \dots, \lambda_{2^{k-1}+n-1}, 0, \dots, 0).$$
    Then we want to match the $n$-th pair. To simplify the notation, we will drop the irrelevant $0$'s in the density and focus on the first and the last pairs:
    \begin{center}
    \begin{tikzcd}
        \diag(\lambda_1, \dots, 0, \dots, \lambda_{2^{k-1}+1}, \dots, 0, \dots) \arrow[d, "TP"]\\
        \diag(\lambda_1, \dots, 0, \dots, 0, \dots, \lambda_{2^{k-1}+n-1}, \dots) \arrow[d, "AD"]\\
        \diag(\lambda_1, \dots, \lambda_{2^{k-1}+n-1}(1-\alpha), \dots, 0, \dots, \lambda_{2^{k-1}+n-1}\alpha, \dots) \arrow[d, "AD"]\\
        \diag(\lambda_1\beta, \dots, \lambda_{2^{k-1}+n-1}(1-\alpha)\beta, \dots, \lambda_1(1-\beta), \dots, \lambda_{2^{k-1}+n-1}\alpha+\lambda_{2^{k-1}+n-1}(1-\alpha)(1-\beta), \dots)\\
    \end{tikzcd}
    \end{center}
    It is clear that 
    \begin{equation*}
    \begin{aligned}
        &\quad\diag(\mu_1, \dots, \mu_n, \dots, \mu_{2^{k-1}+1}, \dots, \mu_{2^{k-1}+n}, \dots)\\
        &= \diag(\lambda_1\beta, \dots, \lambda_{2^{k-1}+n-1}(1-\alpha)\beta, \dots, \lambda_1(1-\beta), \dots, \lambda_{2^{k-1}+n-1}\alpha+\lambda_{2^{k-1}+n-1}(1-\alpha)(1-\beta), \dots)
    \end{aligned}
    \end{equation*}
    has a solution.

    Further, this algorithm preserves the order of the ratios, i.e. $R_1\leq R_2\leq\dots\leq R_n$.
    In particular, 
    \begin{equation*}
    \begin{aligned}
        R_1 &= 1-\beta\\
        R_2 &= 1-\beta+\frac{d_{2^{k-1}+2}}{d_2+d_{2^{k-1}+2}}\alpha\beta\\
        R_k &= 1-\beta+\frac{d_{2^{k-1}+k}}{d_k+d_{2^{k-1}+k}}\alpha\beta\\
        R_n &= 1-\beta+\alpha\beta
    \end{aligned}
    \end{equation*}
    for $1<k<n$. $R_2\leq R_3\leq\dots\leq R_{n-1}$ follows from inductive hypothesis, and $R_{n-1}\leq R_n$ follows from calculation.
    Hence, up to permutation of the diagonals, we obtain the desired diagonal target density.
    
    The above algorithm is optimal within the scheme. The state preparation part, i.e. from arbitrary diagonal to pure, we need $O(2^k)$ many arbitrary transpositions, or $O(2^{2k})$ sparse 2-local transpositions. With both, we need $k$-amplitude dampings on the first register.
    
    For the second part, from pure to arbitrary disgonal, one need $O(2^k)$ many transpositions or $O(2^{2k})$ sparse 2-local transpositions, both with $k$-amplitude dampings on the first register.
\end{example}
Note that the above calculation still holds with a modified resource sets. In particular, in the example above, we only allow amplitude-damping at a single register. One can allow such operation on every single register. This would require a lot more amplitude damping, but will reduce the number of transpositions needed.

If one allow linear combinations of resources, it is not hard to see that state preparation can be done in one step by taking $L_{a_1+a_2+\dots+a_k}$, without need of any transposition. This is essentially doing all $O(k)$ amplitude damping in one go. However, this needs to be taken with extra care due to the obstructions discussed in \cref{obstructions}.

The optimal transport plan for arbitrary density to arbitrary density beyond the toy model with this particular resource set, i.e. sparse transposition and amplitude damping, remains open. There might be a direct path without needing to diagonalize the density first.

In this example, we require minimal application dissipation but significant amount of transposition through Hamiltonians. Conceptually, the Hamiltonians prepare the system in a favorable state, allowing dissipation to act effectively rather than as uncontrollable noise. However, achieving this balance necessitates a carefully designed Hamiltonian scheme. With recent experimental progress in implementing controlled noisy gates, future research should not only aim to identify optimal transport plans but also to maximize the beneficial role of dissipation. This highlights a fundamental trade-off between unitary and dissipative resources: full unitary control may be unnecessary if noise is properly engineered. As laboratory simulations of noise become increasingly feasible and cost-effective, exploiting this trade-off will be of growing practical importance.

\bibliographystyle{alpha}
\bibliography{reference}
\bigbreak

\end{document}